\documentclass[lettersize,journal]{IEEEtran}
\usepackage{amsmath,amsfonts}
\usepackage[linesnumbered, ruled, vlined]{algorithm2e}
\usepackage{amsthm}
\usepackage{amssymb}
\newtheorem{theorem}{Theorem}
\newtheorem{proposition}{Proposition}
\newtheorem{corollary}{Corollary}
\newtheorem{lemma}{Lemma}

\usepackage{array}
\usepackage[caption=false,font=normalsize,labelfont=rm,textfont=rm]{subfig}
\usepackage{textcomp}
\usepackage{stfloats}
\usepackage{url}
\usepackage{verbatim}
\usepackage{graphicx}
\usepackage{cite}
\usepackage{color}
\usepackage[hidelinks]{hyperref}
\def\BibTeX{{\rm B\kern-.05em{\sc i\kern-.025em b}\kern-.08em
    T\kern-.1667em\lower.7ex\hbox{E}\kern-.125emX}}
\usepackage{balance}

\renewcommand{\eqref}[1]{(\ref{#1})}
\newcommand{\figref}[1]{Fig.~\ref{#1}}

\newcommand{\algref}[1]{Alg.~\ref{#1}}
\newcommand{\lineref}[1]{Line~\ref{#1}}
\newcommand{\secref}[1]{Section~\ref{#1}}
\newcommand{\corolref}[1]{Corollary~\ref{#1}}
\newcommand{\propref}[1]{Propositon~\ref{#1}}

\begin{document}
\title{Maximize the Long-term Average Revenue of Network Slice Provider via Admission Control Among Heterogeneous Slices}
\author{Miao Dai, \textit{Member, IEEE}, Gang Sun, \textit{Member, IEEE}, Hongfang Yu, \textit{Member, IEEE}, \\ and Dusit Niyato, \textit{Fellow, IEEE}
\thanks{Miao Dai is with the Key Laboratory of Optical Fiber Sensing and Communications (Ministry of Education), University of Electronic Science and Technology of China, Chengdu 611731, China (e-mail: daimiao@std.uestc.edu.cn).}
\thanks{Gang Sun is with the Key Laboratory of Optical Fiber Sensing and Communications (Ministry of Education), University of Electronic Science and Technology of China, Chengdu 611731, China, and also with Agile and Intelligent Computing Key Laboratory of Sichuan Province, Chengdu 610036, China (e-mail: gangsun@uestc.edu.cn).}
\thanks{Hongfang Yu is with the Key Laboratory of Optical Fiber Sensing and Communications (Ministry of Education), University of Electronic Science and Technology of China, Chengdu 611731, China, and also with Peng Cheng Laboratory, Shenzhen 518066, China (e-mail: yuhf@uestc.edu.cn).}
\thanks{Dusit Niyato is with School of Computer Science and Engineering, Nanyang Technological University, Singapore (e-mail: dniyato@ntu.edu.sg).}}


\maketitle

\begin{abstract}
Network slicing endows 5G/B5G with differentiated and customized capabilities to cope with the proliferation of diversified services, whereas limited physical network resources may not be able to support all service requests. Slice admission control is regarded as an essential means to ensure service quality and service isolation when the network is under burden. Herein, the scenario where rational tenants coexist with partially competitive network slice providers is adopted. We aim to maximize the long-term average revenue of the network operators through slice admission control, with the feasibility of multidimensional resource requirements, the priority differences among heterogeneous slices, and the admission fairness within each slice taken into account concurrently. We prove the intractability of our problem by a reduction from the Multidimensional Knapsack Problem (MKP), and propose a two-stage algorithm called MPSAC to make a sub-optimal solution efficiently. The principle of MPSAC is to split the original problem into two sub-problems; inter-slice decision-making and intra-slice quota allocation, which are solved using a heuristic method and a tailored auction mechanism respectively. Extensive simulations are carried out to demonstrate the efficacy of our algorithm, the results show that the long-term average revenue of ours is at least 9.6\% higher than comparisons while maintaining better priority relations and achieving improved fairness performance.
\end{abstract}


\section{Introduction} \label{sec:introduction}


With the rapid development and continuous iteration of applications, there is an obvious trend of diversification and customization of network services. In the 5G/B5G era, services are broadly classified into three categories; eMBB, URLLC, and mMTC \cite{shafi20175g}, which require the network to support large bandwidth transmission under high mobility, reliable communication with extremely low latency, and massive device connections, respectively. Moreover, each of the three can be further subdivided into a variety of fine services. The difference in service demand is constantly strengthened and deepened with the emergence of new services. Therefore, various industries have put forward unprecedented requirements for differentiated and customized service capabilities of the underlying network \cite{zhang2017network}.

In order to meet the expectations and demands of users in 5G/B5G networks, network slicing has been proposed as an enabling technology to provide flexible customization capabilities for physical networks. It helps network operators to carve out logically independent virtual networks, known as slices, on commonly shared hardware facilities \cite{babbar2022role, abedin2022elastic}, and each slice can be adapted to a specific service type when carefully configured and tailored. This approach brings flexibility and efficacy for network operators to keep up with the evolution of services in a cost-effective manner. A network slice is just a design diagram, and only after being instantiated can it undertake the expected functions. However, maintaining the normal operation of slices requires sufficient resource supply. Otherwise, the virtual network will fail to guarantee the negotiated performance \cite{ko2022pdras}. This means that the network operator will face the financial penalty because of breaking the contract \cite{wu2021reinforcement}. Meanwhile, the quality of experience of users will significantly deteriorate or even completely lost.

The reality is that physical resources are always limited, and the circumstance that the available resources are insufficient to accommodate the aggregate demands of all users happens now and then. To prevent the network from getting overwhelmed by excessive user requests, slice admission control (SAC) is particularly important \cite{ojijo2020survey}. It works as a decision maker that determines which requests can be accepted immediately, and which requests need to be delayed for a while or rejected directly when the network is under a heavy load, thereby ensuring the performance of slices, avoiding intensifying resource competition and damaging the isolation and stability of services. In general, the network operators are responsible for the slice admission control, because they are the owners of the infrastructure and the executors of the slicing operations. They act as network slice providers (NSPs) to earn revenue by renting out slices to vertical service providers (VSPs), the latter use authorized slices to deliver specific services to the corresponding subscribers.

During the slice admission control procedure, resource feasibility is always the primary consideration, as it is closely related to whether the NSP can successfully obtain rental income. However, the resources consumed to run a slice are usually of multiple dimensions, mainly including communication resources (e.g., bandwidth) and computing resources (e.g., CPU and memory), and are needed to be dealt with simultaneously, which makes SAC difficult with the addition of service diversity and heterogeneity. Furthermore, there are natural priority differences among various services, and the slice admission control process may respect this difference to stimulate the capability of the network to make it more suitable for heterogeneous service performance requirements \cite{hossain2021priority}. Meanwhile, since the substrate network is a shared infrastructure, absolute priority is not advisable in that slices with lower priorities need to be supported by NSPs to prevent their subscribers from completely losing service opportunities \cite{haque20225g}. In addition, the rational behavior of slice tenants may affect the revenue of NSPs, which is especially prominent in the scenario of multiple NSPs and multiple VSPs \cite{boateng2022consortium, guan2018demand}. The competition among the participants is both a challenge and an opportunity \cite{wu2021distributed}. If an NSP can make an optimal admission decision, it will attract more slice tenants to establish cooperation with it, thereby expanding its potential revenue space. Otherwise, it will confront the dilemma that some of its tenants may gradually leave to the competitors.

In this work, we aim to propose an effective slice admission control method for a Multi-NSP-Multi-VSP (MNMV) scenario. It differs from the literature in that most of them only consider a single NSP, so that there is no competition and the market is relatively stable, which is no longer the case now. Our objective is to optimize the long-term average revenue of NSPs at the premise of guaranteeing the multidimensional resource demands of each slice, satisfying the inter-slice priority relations, and improving the intra-slice fairness. A confluence of these factors further complicates this issue, but existing works still lack attention to it. The main contributions are summarized below:

\begin{enumerate}
	\item We introduced a scenario where multiple NSPs and VSPs coexist to narrow the gap between the literature and reality. Herein, the slice sets of different NSPs overlap each other and VSPs are assumed rational. It means that there are competitions among peer entities and interactions between slice sellers and tenants can be strategic. The impatient behaviors of subscribers are also taken into account, which makes the scenario more practical and worth investigating.
	
	\item We have established a mathematical model for the online slice admission control problem with the goal of maximizing long-term average revenue, accounting for meeting the resource requirements of heterogeneous slices, maintaining the priority relationship across slices, and achieving fairness in the allocation of admission quotas within slices. Through a reduction from the multidimensional knapsack problem, we proved the NP-Hard property of our problem.
	
	\item Considering the intractability of finding the optimal solution, we designed an algorithm called Multi-participant Slice Admission Control (MPSAC) that resolves the problem in two stages to get a sub-optimal solution efficiently. The first stage handles the inter-slice admission decisions heuristically, and the second stage tackles the intra-slice quota allocation relying on a tailored single-parameter auction mechanism. The priority relation across slices is maintained in the first stage, and the fairness of quota allocation within each slice is promised by the auction without altering the inter-slice decisions.
	
	\item To demonstrate the advantages of MPSAC, we carried out extensive simulations to compare the performance of the proposed scheme with the state of the arts. Results show that our algorithm outperforms other baselines in terms of long-term average revenue, priority maintenance, and intra-slice fairness.
\end{enumerate}

The rest of the paper is organized as follows. \secref{sec:related works} introduces the work of predecessors and discusses their main ideas and limitations. \secref{sec:system architecture} elaborates the system architecture and formulates the mathematical model. In \secref{sec:algorithm design}, the intractability of the problem is analyzed theoretically and the algorithm is devised subsequently. \secref{sec:numerical simulations} presents simulation settings and evaluates the performance of algorithms. Finally, we conclude the paper and raise thinking about future works in \secref{sec:conclusion and future works}.

\section{Related Works} \label{sec:related works}



Admission control is a concomitant problem of the network. Traditionally, network operators have only needed to provide approximate resource provisioning for services in a one-size-fits-all manner. Nowadays, the application of network slicing makes it develop into slice admission control (SAC), and puts forward higher requirements for network operators, such as providing differentiated QoS guarantees for heterogeneous slices, complying with slice priority differences, and maintaining certain fairness of resource scheduling.

\subsection{QoS Guarantees}

The fundamental principle of SAC is to ensure the satisfaction of the user's QoS requirements, which is inseparable from accurate performance abstraction and strict resource guarantee. The most prevalent approach is to first map QoS to concrete resource demands, as advocated in \cite{guo2019enabling, fossati2020multi}, and then defer or filter out excessive requests that lead to resource deficits if accepted. On the contrary, some researchers believe that converting QoS into resource demands is not stable enough. For example, the authors of \cite{kak2021towards} mentioned that the resource overhead of the same slice may be different on various infrastructures, so they stuck with the original Service Level Agreement (SLA) to characterize slices. On the basis of ensuring the QoS requirements, network operators usually tend to shift their focus on SAC to exploring the possibility of providing additional features to cater to the market or respond to government orders.

\subsection{Priority Differences}

The emergence of network slicing has created the paradigm of Slice-as-a-Service (SlaaS), where slices are tightly bound to services and reasonably inherit their priorities. This phenomenon has been captured and taken into account when conducting slice admission control and resource scheduling by some existing works. The authors of \cite{ginthor20215g, yarkina2022multi, yarkina2022performance} believe that high-priority slices have the right to snatch resources from low-priority ones. Wherein, paper \cite{ginthor20215g} schedules physical resource blocks (PRB) for deterministic traffic (always more privileged) and eMBB traffic, the former is allocated a part of PRBs and allowed to preempt a fraction of PRBs reserved for eMBB traffic when the margin is insufficient. While the priority of tenants is adjusted dynamically in \cite{yarkina2022multi, yarkina2022performance}; When the number of users admitted to the tenant does not reach the number agreed in the contract, its priority keeps 1 or gradually decreases otherwise. SAC based on iterative convex optimization is carried out in \cite{yarkina2022multi} to permit high-priority tenants to seize the quotas of low-priority ones that exceed the contracted number, it is further enhanced with machine learning method in \cite{yarkina2022performance} for faster solutions to make it suitable for RAN.

From the perspective of users, preemptive behavior may lead to serious QoE degradation for vulnerable parties, thereby more researchers have adopted a gentle approach to maintaining priority. The easiest way is to access slice requests strictly in order of priority, which is what \cite{guo2019enabling} does. It introduces the idea of dual priority by using the virtual slice that accommodates one or multiple actual slices, the quotas are first allocated in the principle of the earliest deadline first among virtual slices and the internal order is based on the priority given by the service type. A more flexible approach is to impose the priority directly or indirectly as a weighting factor on the optimization goal, which refers to revenue in \cite{haque20225g}, number of collision-free preambles in \cite{gedikli2022deep} and flow acceptance ratio in \cite{li2023slice}. Another benefit of this approach is that it blends naturally with existing optimization methods, e.g., deep Q-learning, actor-critic paradigm, and heuristic one employed in \cite{haque20225g, gedikli2022deep, li2023slice} respectively. The authors of \cite{abedin2022elastic} have made a combination of the two. They manage to minimize the aggregate Age of Information (AoI) of terminal equipment in the Industrial Internet of Things. Service priorities there constitute the preference list for base stations to associate users and are converted into weights of the corresponding AoI so that each base station can treat the users connected to it differently when allocating PRBs.

\subsection{Fairness Considerations}

The research on the priority of SAC also arouses the thinking about fairness. The authors of \cite{seah2021combined} use the \textit{Jain}'s fairness index \cite{jain1984quantitative} defined on the bandwidth supply-demand ratio to measure the fairness of resource scheduling among the three types of slices. The potential limitation is that a consistent bandwidth supply-demand ratio does not necessarily mean that users can obtain similar QoE, because the QoS sensitivity of these slices to a single resource type may be quite different. Therefore, there are considerable researches \cite{wang2021inter, modina2022multi} that prefer to use indicators closer to services and users to measure the performance of heterogeneous slices. The authors of \cite{wang2021inter} think that the SLA violation ratio is a possible choice. The metric becomes utility in \cite{modina2022multi}, where the infrastructure provider (InP) lends multidimensional resources to service providers to promote the \textit{$\alpha$-fairness} of both the overall utility of each service provider and the utility obtained at different service locations of a provider. A general framework with Ordered Weighted Averaging (OWA) operator is proposed in \cite{fossati2020multi}, it assigns weights to elements according to their quality when aggregating them, instead of always assigning a fixed weight to an element. By adjusting the weights, this framework can incorporate several well-known criteria for fairness measurement. And different slice satisfaction indicators can also be supported as long as the elements can be properly defined and calculated.

The works above do not take slice priorities into account, possibly influenced by the intuitive contradiction of the two concepts. Other researchers have a more open mind, they argue that fairness and priority are not opposite, but unified with each other. The authors of \cite{ogryczak2014fair} believe that fairness does not imply equal treatment, and those of \cite{caballero2017multi} think that a decision can only be considered fair if it takes into account the contribution variance of all entities. In \cite{nasser2004optimal}, a lower handoff dropping probability is provided for high-priority requests, and the gaps of which between adjacent priorities are controlled under a specific threshold to reflect fairness. We have also adopted a similar idea in our previous work \cite{dai2022psaccf}, where high-priority slices are guaranteed higher cumulative acceptance ratios, and fairness is pursued by improving the evenness of all cumulative acceptance ratio gaps between adjacent priorities. It is undeniable that the consideration of priority and fairness enriches the significance of SAC, making it more comprehensive and practical.

\subsection{Strategic Behaviors}

In addition to the above factors related to slices, some works also consider the specificity of the market, such as the strategic behavior of tenants and the impatient behavior of users. The literature \cite{caballero2018network} and \cite{caballero2019network} assume that tenants are rational, they can know the total request information of other tenants, so as to adjust their request decisions to maximize their utility. The authors deduce the formula of the best response dynamics, and assert in \cite{caballero2018network} that when tenants have data rate requirements, the game can converge to a region near the Nash Equilibrium (NE) with the existence of SAC, otherwise, it can directly converge to the NE in \cite{caballero2019network}. The authors of \cite{boateng2022consortium} assume that both InPs and tenants are rational. InPs determine their prices of the RAN spectrum and tenants decide the demands, both parties aim to optimize their respective utility. It is modeled as a Stackelberg game with multiple leaders and multiple followers, and a multi-agent deep reinforcement learning (MADRL) method is proposed to find a Stackelberg Equilibrium (SE). The impatient behavior is reflected in \cite{han2020multiservice}; tenants can give up initiating a new slice request or withdraw a request that has already been made when foreseeing or having already waited for a long waiting time. For this reason, the authors propose a SAC paradigm based on heterogeneous queues to fully schedule available resources and improve the slice provider's revenue.

The above works have conducted unique and valuable research on network slicing. However, to the best of our knowledge, there is no work on SAC in the MNMV market that simultaneously considers the multidimensional resource feasibility, slice priority differences, and intra-slice fairness while assuming both rational tenants and impatient users. It is a more sophisticated and realistic issue that dramatically increases the challenges network operators confront in pursuit of their goals (e.g., revenue, utility). Therefore, it deserves attention and in-depth investigation.

\section{System Architecture} \label{sec:system architecture}
In this work, we consider a more realistic scenario where multiple Network Slice Providers (NSPs) and Vertical Service Providers (VSPs) perform their duties to deliver QoS-guaranteed services to end users. \figref{fig:scenario} shows the details.

\begin{figure}[!tbh]
	\centering
	\includegraphics[scale=0.55]{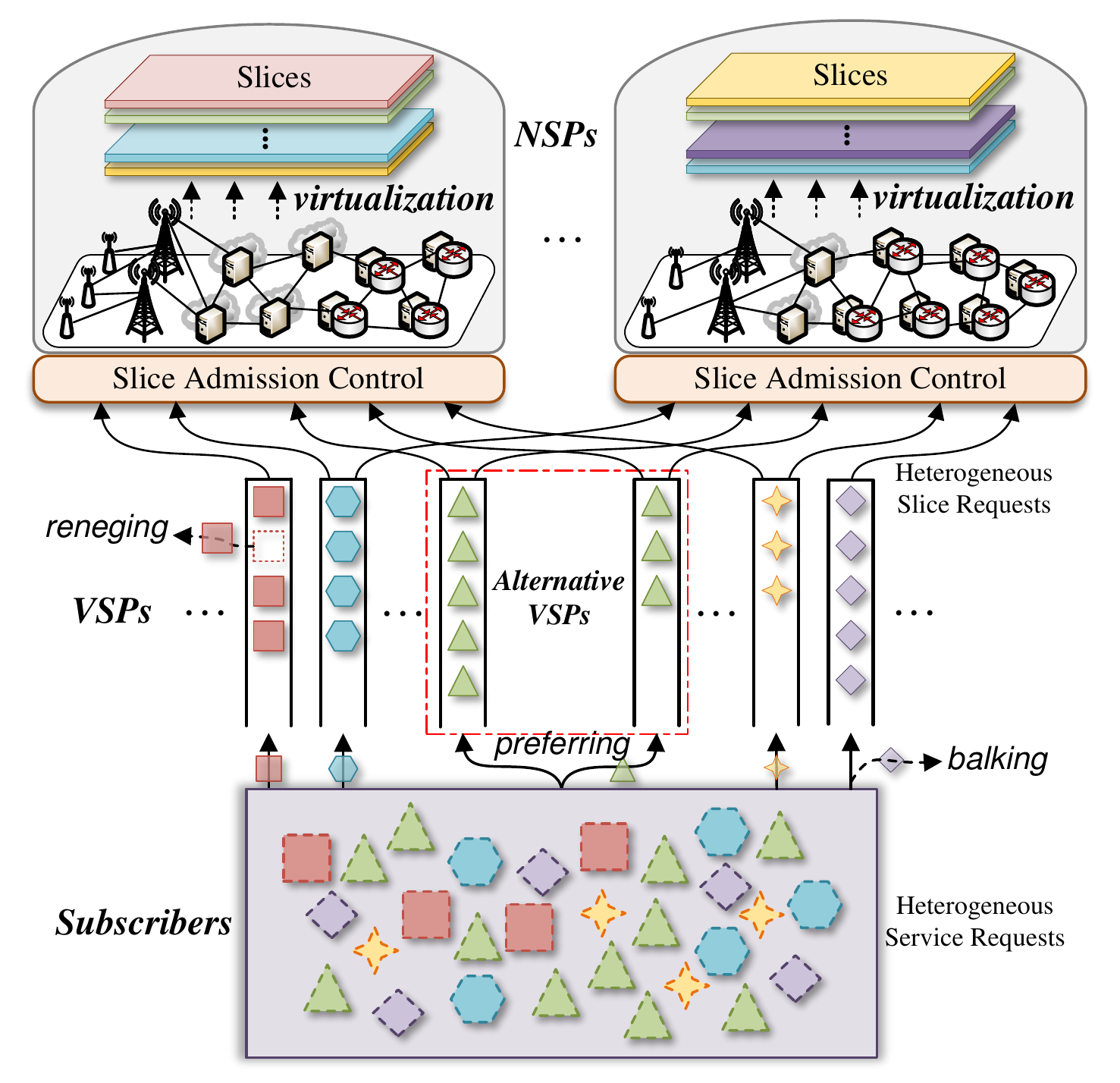}
	\caption{The scenario with multiple NSPs and heterogeneous VSPs.}
	\label{fig:scenario}
\end{figure}

\subsection{Participants}
We consider three basic market participants here \cite{sexton2020provisioning}:
\begin{enumerate}
	\item \textit{Subscribers}. They are the applicants and users of the services. Generally, subscriptions are initiated by end users through terminal devices with different levels of quality and key performance indicators.
	
	\item \textit{VSPs}. They are companies or organizations that provide customized Internet services in specific industries. Most of them build their own app, website, or client to deliver purposely designed, developed, and maintained service content or data to their subscribers.
	
	\item \textit{NSPs}. They are the owners of the substrate network. They abstract logically independent virtual networks, namely slices, from generic hardware with the help of SDN/NFV technologies. By instantiating network slices, NSPs provide the carrier of service entities to VSPs and work as the actual bearer of user traffic.
	
	Due to the differences in types and deployment locations of network equipment held by distinct NSPs, the slice categories that they support can partially overlap \cite{li2023slice}.
\end{enumerate}

\subsection{Interaction Pattern}
\figref{fig:scenario} describes the interaction process between the three parties. Subscribers first initiate requests to corresponding VSPs. VSPs maintain their respective queues to hold the received service requests and then determine the number of slices to apply for each optional NSP that can implement their service entities. NSPs gather all slice requests, and decide the access quota granted to each VSP considering resource feasibility and some other concerns (e.g. priority and fairness). After obtaining slice authorizations, VSPs establish their service entities and serve the waiting requests in a First Come First Serve (FCFS) manner. Upon a subscriber ends an ongoing service, the VSP serving it immediately raises a termination message to the NSP. Then, the NSP withdraws the slice permission and recycles the resources released.

During the interaction, participants are assumed rational in that they will observe and roughly analyze the interaction state to make decisions that are most beneficial to them.

\subsubsection{For subscribers}
There are mainly three possible behaviors: \textit{reneging}, \textit{balking} \cite{wang2010queueing} and \textit{preferring}.
\begin{itemize}
	\item \textit{Reneging}. A VSP may not be able to meet service requests immediately upon receipt, but put them in the queue in the order of arrival and wait for subsequent scheduling.
	
	However, waiting too long can make subscribers cancel the corresponding request and retreat from the queue. This behavior is called \textit{reneging} and can be simulated by setting a maximum tolerable period for each service request.
	
	\item \textit{Balking}. It refers to the behavior that a subscriber eventually gives up submitting a planned request due to a non-empty queue. When a subscriber needs service, it first observes the backlog of requests in the corresponding queue. The longer the queue, the less motivated the subscriber is to announce a new request because it is easier to renege after entering the queue. A probability function that monotonically decreases with the queue length is a commonly adopted approach to characterize this behavior.
	
	Note that the queue length is private information to the VSP, but we assume that the VSP will reveal it to the subscribers. This is believed to be a win-win for VSPs and subscribers; queue status can help subscribers make more informed choices, thereby avoiding frequent reneging that deteriorates VSP's reputation and user enthusiasm \cite{han2020multiservice}.

	\item \textit{Preferring}. In a scenario with multiple participants, there may exist several independent VSPs that provide the same kind of service, i.e., the \textit{Alternative VSPs} (AVSPs) depicted in \figref{fig:scenario}. In this case, subscribers of the service can choose any VSP with equal probability if they have the same queue length, but always prefer the one with the shortest queue otherwise.
	
	It is worth mentioning that we exclude the influence of service QoS and price on user choice. In the SlaaS paradigm, the same slice instance means that they provide the same QoS. Moreover, in a competitive environment, prices are mainly determined by the market \cite{buttyan2007security, tran2020resource} and can also be considered consistent.
\end{itemize}
 
\subsubsection{For VSPs}
In order to provide services to subscribers, VSPs need to rent physical resources from NSPs in units of slice instances. Some services are only suitable for a specific NSP, others are available from multiple NSPs.

When several NSPs can provide slices required by a VSP, the VSP will evaluate the service capabilities of these NSPs before initiating slice requests. Specifically, the VSP compares slice acceptance ratio and inter-slice admission fairness of them, and then announces more slice requests to the NSP with higher two metrics, because such NSP is believed to promise more admission opportunities stably.

Without loss of generality, we assume that NSPs are willing to publish this information, as it is an important way for them to attract tenants from other competitors.

\subsubsection{For NSPs}
The purpose of NSP is to earn revenue by leasing slice instances to VSPs. One constraint that must be considered in advance is to guarantee sufficient resources to each slice instance. Resource deficit cannot ensure the quality of slices to meet the SLA. In this case, NSP suffers from penalties. Therefore, admitting concurrent slice requests that exceed NSP's resource-carrying capability will not yield more revenue, but serious losses instead.

As in previous work \cite{haque20225g, li2023slice}, we require NSPs to prioritize heterogeneous slices, but at the same time, they are additionally expected to balance the admission quotas across and within slices. There are mainly two reasons; as the owners of the underlying network, NSPs are responsible for provisioning slices to the general public, which means that serving high-priority slices at the cost of starving low-priority ones is unacceptable. Moreover, when multiple NSPs coexist, the ability to retain the priority relationship and balance the admission quantities helps to attract more tenants from competitors, so as to obtain a wider source of revenue.


\subsection{Mathematical Model}
There are $N$ NSPs and we denote them with a set $\mathcal{N}=\{1,2,\dots,N\}$, NSP~$n$ can support $|\mathcal{S}_n|$ kinds of slices, where $\mathcal{S}_n$ is the set of slice labels. We use slice labels to indicate priority, bigger labels mean higher priorities and all participants have a consensus on the relative priority relationship between slices \cite{debbabi2022overview}. We follow the approach in \cite{guo2019enabling, fossati2020multi} to use resource overhead vectors to describe the QoS requirements of slice instances. Different slices will consume different amounts of multidimensional resources, let $\boldsymbol{c}_{n,s}=(c^1_{n,s}, c^2_{n,s}, \dots, c^K_{n,s})^T$ denote the column vector of minimum quantities that NSP~$n$ should spend on each $s$-type slice instance for the sake of guaranteeing QoS. The overall resource capacities of NSP~$n$ is $\boldsymbol{C}_n=(C^1_n, C^2_n, \dots, C^K_n)^T$, where $K$ indicates the total resource types. NSP~$n$ charges at least $p_{n,s}$ for each active slice instance of type~$s$ per time slot, we name it the base price. Since it is widely accepted that the price of a slice is primarily driven by market factors \cite{buttyan2007security, tran2020resource}, we assume if several NSPs provide the same slice, they adopt a unified base price for the slice and all the prices are clearly marked and disclosed.

Let $\mathcal{V}=\{1,2,\dots,V\}$ denote the set of VSPs. The optional NSPs that VSP~$v$ can choose constitute a subset $\mathcal{N}_v \subseteq \mathcal{N}$, and VSPs that ask for slice~$s$ are included in a subset $\mathcal{V}_s \subseteq \mathcal{V}$. Each VSP has a private estimation of the value of a slice instance, which may vary due to their different operational expenditures when the prices of alternative services are (assumed) the same. Let $b^*_v$ represent the valuation of VSP~$v$ about the slice~$s_v$, i.e., the label of slice that supports VSP~$v$.

We employ a slotted admission process. Slot~$t$ starts with the procedure of death and birth. Active slice instances that have reached their maximum lifetime expire and the resources occupied are withdrawn immediately, and those that hang up in the queue also depart if their maximum waiting time runs out. We denote the maximum lifetime and waiting time of a slice/service instance of type~$s$ as $T^L_s$ and $T^W_s$ respectively, which obey \textit{Exponential} distribution \cite{ancker1963some} with parameter $1/\lambda^L_s$ and $1/\lambda^W_s$. Meanwhile, subscribers generate new service requests according to \textit{Poisson} distribution with parameter $\lambda^G_s$, they first complete the \textit{preferring} step to select a VSP, say $v$, and then enter the queue with the probability given by \eqref{eq:balking}:
\begin{equation}
	\label{eq:balking}
	P_v = e^{-\beta_v\cdot l_v},
\end{equation}

\noindent where $\beta_v\in[0,1]$ is the willingness that a subscriber of VSP~$v$ to wait and $l_v$ is the real-time queue length. We add a superscript $t$ to $l_v$, i.e., $l^t_v$, to denote the queue length of VSP~$v$ after the death and birth procedure in time slot~$t$.

Then, VSP~$v$ checks if there are any requests in its queue. If so, it manages to rent slice instances from all optional NSPs. When only one NSP is available, VSP~$v$ just matches each service request in the queue with a slice request sent to the NSP. Otherwise, VSP~$v$ strategically determines the number of requests to each optional NSP~$n$ according to \eqref{eq:request prop}:
\begin{equation}
	\label{eq:request prop}
	l^t_{v,n}=\left[\frac{\alpha \cdot e^{r^{t-1}_{n,s_v}}}{\sum_{n^\prime \in \mathcal{N}_v}{e^{r^{t-1}_{n^\prime,s_v}}}} + 
	\frac{(1-\alpha)\cdot e^{f^{t-1}_n}}{\sum_{n^\prime \in \mathcal{N}_v}{e^{f^{t-1}_{n^\prime}}}}\right]\cdot l^t_v, \forall n \in \mathcal{N}_v,
\end{equation}

\noindent where $\alpha\in[0,1]$ is the weighting factor. $r^t_{n,s}$ is the cumulative acceptance ratio of NSP~$n$ on slice~$s$ at the end of slot~$t$, given by \eqref{eq:cumulative acceptance ratio}. $f^t_n$ is the inter-slice admission balance of NSP~$n$ at the end of slot~$t$, given by \eqref{eq:inter slice fairness}.
\begin{equation}
	\label{eq:cumulative acceptance ratio}
	r^t_{n,s} = \frac{\sum_{\tau=1}^{t}{\sum_{v\in\mathcal{V}_s}a^\tau_{n,v}}}{\sum_{\tau=1}^{t}{\sum_{v\in\mathcal{V}_s}l^{\tau}_{v,n}}},
\end{equation}
\begin{equation}
	\label{eq:inter slice fairness}
	f^t_n = 
	\begin{cases}
		\frac{\left[\sum_{s\in\mathcal{S}_n}{\left(r^t_{s^\prime}-r^t_s\right)}\right]^2}{\left(|\mathcal{S}_n|-1\right)\cdot\sum_{s\in\mathcal{S}_n}{\left(r^t_{s^\prime}-r^t_s\right)^2}}, & \text{if } r^t_{n,s}\leq r^t_{n,s^\prime},\\
		0, & \text{otherwise.}
	\end{cases},
\end{equation}

Specifically, $a^\tau_{n,v}$ in \eqref{eq:cumulative acceptance ratio} denotes the number of slice quotas that NSP~$n$ grants to VSP~$v$ at time slot~$\tau$, thereby the numerator (denominator) of \eqref{eq:cumulative acceptance ratio} indicates the number of slice request of type~$s$ that NSP~$n$ has cumulatively served (received) up to slot~$t$. While $s^\prime$ in \eqref{eq:inter slice fairness} denotes the smallest label in set~$\mathcal{S}_n$ that is bigger than $s$. Therefore, \eqref{eq:inter slice fairness} means that if the cumulative acceptance ratio of slices does not increase with their priorities, NSP~$n$ is considered (by VSPs) to be totally failed to maintain admission balance among slices at slot~$t$. Otherwise, NSP~$n$ succeeds and the degree of balance is measured by the \textit{Jain}'s fairness index \cite{jain1984quantitative} defined on the gaps of adjacent cumulative acceptance ratios. \eqref{eq:inter slice fairness} is proposed to measure fairness in our previous work \cite{dai2022psaccf}, we continue to use it here and give it a more rigorous name; inter-slice fairness.

Actually, \eqref{eq:request prop} formulates the principle VSPs adopt when selecting NSPs; they prefer those with higher cumulative acceptance ratio and willing to respect slice priority, or more precisely, with better inter-slice fairness, and hence propose more slice requests to such NSPs. Intuitively, a higher cumulative acceptance ratio indicates that slice requests are more likely to be satisfied at this NSP, and better inter-slice fairness can ensure that even if there are a large number of high-priority requests from other VSPs, those with lower priority will not completely lost opportunities, i.e., a better service balance is promised.


When NSPs have received all slice requests from VSPs, they make admission decisions with the goal of maximizing long-term average base revenue. To do this, they must guarantee the resource feasibility to avoid penalties and try to preserve priority differences to absorb more tenants. (P1) shows the mathematical model of the optimization problem.
\begin{align}
	\text{(P1): } \underset{a^t_{n,v}}{\max} \underset{T \rightarrow  \infty}{\text{lim}}\frac{1}{T}\sum_{t=1}^{T}{\sum_{s \in \mathcal{S}_n}{p_{n,s} \cdot \left(A^t_{n,s}+\sum_{v \in \mathcal{V}_s}{a^t_{n,v}}\right)}} \label{eq:obj 1} \\
	\text{s.t. } \sum_{s \in \mathcal{S}_n}{\boldsymbol{c}_{n,s} \cdot \left(A^t_{n,s}+\sum_{v \in \mathcal{V}_s}{a^t_{n,v}}\right)} \leq \boldsymbol{C}_n, \forall t, \tag{\ref{eq:obj 1}a} \label{eq:obj 1a} \\
	a^t_{n,v} \leq l^t_{v,n}, \forall v \in \mathcal{V}_n, \forall t, \tag{\ref{eq:obj 1}b} \label{eq:obj 1b} \\
	a^t_{n,v} \geq 0, \forall v \in \mathcal{V}_n, \forall t, \tag{\ref{eq:obj 1}c} \label{eq:obj 1c} \\
	a^t_{n,v} \in \mathbb{Z}, \forall v \in \mathcal{V}_n, \forall t, \tag{\ref{eq:obj 1}d} \label{eq:obj 1d} \\
	\text{(SC) } r^t_{n,s} \leq r^t_{n,s^\prime}, \forall s,s^\prime \in \mathcal{S}_n, \forall t, \tag{\ref{eq:obj 1}e} \label{eq:obj 1e}
\end{align}

The optimization object of (P1) is the long-term average base revenue of NSP~$n$. The decision variables $\left\{a^t_{n,v}|v \in \mathcal{V}_n \right\}$ are the number of slice instances that should be additionally leased to VSP~$v$ at slot~$t$. $A^t_{n,s}$ in \eqref{eq:obj 1} refers to the number of slices of type~$s$ that were admitted by NSP~$n$ previously and still keep alive during slot~$t$. NSP~$n$ allocates $\boldsymbol{c}_{n,s}$ amount of physical resources to each admitted $s$-type slice request to ensure that the slice instance can meet the QoS requirement. Therefore, \eqref{eq:obj 1a} corresponds to the resource feasibility constraint of NSP~$n$; the aggregate resources occupied by active and newly admitted slice instances at slot~$t$ cannot exceed its resource capacity. \eqref{eq:obj 1b}-\eqref{eq:obj 1d} restrict the decision variables to be non-negative integers and up-bounded by the requested volume of the corresponding VSP. \eqref{eq:obj 1a}-\eqref{eq:obj 1d} are hard constraints, which must be strictly satisfied at any time.

Besides, \eqref{eq:obj 1e} gives another special constraint, it conveys the vision that the cumulative acceptance ratio of a high-priority slice is not less than that of a low-priority one. However, due to the existence of the above hard constraints, there may not always be a point in the feasible region of \eqref{eq:obj 1a}-\eqref{eq:obj 1d} that also meets the last constraint. Therefore, we introduce \eqref{eq:obj 1e} as a soft constraint (SC) to emphasize differences in slice priorities. A soft constraint is not a necessary condition, but one to be satisfied when it is possible.

Note that (P1) does not address the issue of intra-slice fairness, we defer it to \secref{subsec:intra-slice quota auction}.

\section{Algorithm Design} \label{sec:algorithm design}
The model formulated in (P1) is essentially an online decision problem, which means that each NSP can only determine the variables of the current slot, rather than altering those committed previously or those not yet present. It is worth noting that (P1) is even more complex than ordinary online problems, because the decision at slot~$t$ will affect the subsequent selections of VSPs, thereby influencing the value of some constants in the model in future slots.

In this section, we first demonstrate the intractability of (P1) through its per-slot version with a reduction from the Multi-dimensional Knapsack Problem (MKP), which is known to be NP-Hard \cite{freville2004multidimensional}. Then, we propose a two-stage SAC algorithm to get a sub-optimal solution in polynomial time. The first stage solves the problem of inter-slice admission based on Dominant Resource Revenue Efficiency (DRRE), and the second stage completes the intra-slice quota allocation with a well-designed auction mechanism.

By the way, although some researchers have explored detailed resource allocation schemes, we pay more attention to SAC itself in this work. Therefore, resources are allocated to each admitted slice request according to the resource overhead of maintaining the normal operation of the slice instance.

\subsection{Theoretical Analysis}
To the best of our knowledge, there is no existing method that can optimally solve the online MINLP like (P1) in an time-efficient way. It is widely adopted to optimize the objective for each slot independently, which is illustrated in \cite{salvat2018overbooking} to have the ability to well approximate the original one when the system time horizon is relatively larger than slot. We employ this greedy idea to design our algorithm at the cost of appropriately sacrificing the optimality of the solution. (P2) gives the per-slot version of (P1) after some equivalent transformations. (P2) is believed less intractable than (P1) because it neither considers the impact of current decisions on the follow-up, nor involves unknown information in the future \cite{challa2019network}.
\begin{align}
	\text{(P2): } \underset{a^t_{n,v}}{\max} \sum_{s \in \mathcal{S}_n}{\left(p_{n,s} \cdot \sum_{v \in \mathcal{V}_s}{a^t_{n,v}}\right)} \label{eq:obj 2} \\
	\text{s.t. } \sum_{s \in \mathcal{S}_n}{\left(\boldsymbol{c}_{n,s} \cdot \sum_{v \in \mathcal{V}_s}{a^t_{n,v}}\right)} \leq \boldsymbol{C}_n - \sum_{s \in \mathcal{S}_n}{\boldsymbol{c}_{n,s} \cdot A^t_{n,s}}, \tag{\ref{eq:obj 2}a} \label{eq:obj 2a} \\
	a^t_{n,v} \leq l^t_{v,n}, \forall v \in \mathcal{V}_n, \tag{\ref{eq:obj 2}b} \label{eq:obj 2b} \\
	a^t_{n,v} \geq 0, \forall v \in \mathcal{V}_n, \tag{\ref{eq:obj 2}c} \label{eq:obj 2c} \\
	a^t_{n,v} \in \mathbb{Z}, \forall v \in \mathcal{V}_n, \tag{\ref{eq:obj 2}d} \label{eq:obj 2d} \\
	\text{(SC) } r^t_{n,s} \leq r^t_{n,s^\prime}, \forall s,s^\prime \in \mathcal{S}_n, \tag{\ref{eq:obj 2}e} \label{eq:obj 2e}
\end{align}

\begin{theorem} \label{theorem:1}
	(P2) is NP-Hard.
\end{theorem}
\begin{proof}
	Let us consider the multidimensional knapsack problem in \eqref{eq:MKP}, where $v_i$ and $\boldsymbol{w}_i$ are the value and volume vector of item~$i$ respectively, and $\mathcal{I}$ is the set of items. \eqref{eq:MKPa} is equivalent to \eqref{eq:obj 2a} by mapping the right hand side of \eqref{eq:obj 2a} to $\boldsymbol{W}$, $\left\{\boldsymbol{c}_{n,s}|s \in \mathcal{S}_n \right\}$ to $\left\{\boldsymbol{w}_i|i \in \mathcal{I} \right\}$, and $\left\{\sum_{v \in \mathcal{V}_s}{a^t_{n,v}}|s \in \mathcal{S}_n \right\}$ to $\left\{x_i|i \in \mathcal{I} \right\}$. \eqref{eq:MKPb} can be realized by combining \eqref{eq:obj 2b}-\eqref{eq:obj 2d} when each $l^t_{v,n}$ is set to 1. The two objectives depicted in \eqref{eq:obj 2} and \eqref{eq:MKP} are also consistent when we map $\left\{p_{n,s}|s \in \mathcal{S}_n \right\}$ to $\left\{v_i|i \in \mathcal{I} \right\}$. \eqref{eq:obj 2e} is a soft constraint, ignoring it makes (P2) easier and the simplified version is exactly MKP after the above mappings.
	\begin{align}
		\text{MKP: } \underset{x_i}{\max} \sum_{i \in \mathcal{I}}{v_i \cdot x_i} \label{eq:MKP} \\
		\text{s.t. } \sum_{i \in \mathcal{I}}{\boldsymbol{w}_i \cdot x_i} \leq \boldsymbol{W}, \tag{\ref{eq:MKP}a} \label{eq:MKPa} \\
		x_i \in \{0,1\}, \forall i \in \mathcal{I}, \tag{\ref{eq:MKP}b} \label{eq:MKPb}
	\end{align}
	
	Therefore, any problem that belongs to MKP, it can be solved by an algorithm designed for (P2) with \eqref{eq:obj 2e} abandoned. In other words, MKP can reduce to a simplified version of (P2). Since MKP is a well-known NP-Hard problem, so does (P2).
\end{proof}

\begin{corollary} \label{corollary:1}
	(P1) is NP-Hard.
\end{corollary}
\begin{proof}
	Considering that (P1) is more intractable than (P2), and (P2) is proved NP-Hard in Theorem~\ref{theorem:1}, we can easily conclude that (P1) is also NP-Hard.
\end{proof}

Slice admission control must be performed quickly and efficiently to meet the timeliness requirements of services. However, it is a consensus that NP-Hard problems cannot be solved in polynomial time unless suboptimal solutions are acceptable. Thus, our idea of solving (P2) is to give up searching for the optimal solution and to find a feasible solution with a better objective value.

We substitute $\sum_{v \in \mathcal{V}_s}{a^t_{n,v}}$ in (P2) with a single variable $a^t_{n,s}$, which indicates all the quotas that NSP~$n$ gives to slice~$s$, and sum \eqref{eq:obj 2b}-\eqref{eq:obj 2d} over $v \in \mathcal{V}_s$. These operations construct a new problem (P3), which has fewer variables than (P2).
\begin{align}
	\text{(P3): } \underset{a^t_{n,s}}{\max}\sum_{s \in \mathcal{S}_n}{p_{n,s} \cdot a^t_{n,s}} \label{eq:obj 3} \\
	\text{s.t. } \sum_{s \in \mathcal{S}_n}{\boldsymbol{c}_{n,s} \cdot \left(A^t_{n,s}+a^t_{n,s}\right)} \leq \boldsymbol{C}_n, \tag{\ref{eq:obj 3}a} \label{eq:obj 3a} \\
	a^t_{n,s} \leq \sum_{v \in \mathcal{V}_s}{l^t_{v,n}}, \forall s \in \mathcal{S}_n, \tag{\ref{eq:obj 3}b} \label{eq:obj 3b} \\
	a^t_{n,s} \geq 0, \forall s \in \mathcal{S}_n, \tag{\ref{eq:obj 3}c} \label{eq:obj 3c} \\
	a^t_{n,s} \in \mathbb{Z}, \forall s \in \mathcal{S}_n, \tag{\ref{eq:obj 3}d} \label{eq:obj 3d} \\
	\text{(SC) } r^t_{n,s} \leq r^t_{n,s^\prime}, \forall s,s^\prime \in \mathcal{S}_n, \tag{\ref{eq:obj 3}e} \label{eq:obj 3e}
\end{align}

\begin{lemma} \label{lemma:1}
	Any feasible solution of (P3) can be mapped to a feasible solution of (P2) in polynomial time, and both have the same objective value.
\end{lemma}
\begin{proof}
	To equalize the two objectives in \eqref{eq:obj 2} and \eqref{eq:obj 3}, it requires that \eqref{eq:two obj equal} holds. Moreover, any $\left\{a^t_{n,s}|s \in \mathcal{S}_n \right\}$ that satisfies \eqref{eq:obj 3a} and \eqref{eq:obj 3e} can promise $\left\{a^t_{n,v}|v \in \mathcal{V}_s, s \in \mathcal{S}_n \right\}$ satisfying \eqref{eq:two obj equal} satisfy \eqref{eq:obj 2a} and \eqref{eq:obj 2e}.
	\begin{equation}
		\sum_{v \in \mathcal{V}_s}{a^t_{n,v}} = a^t_{n,s}, \forall s\in \mathcal{S}_n, \label{eq:two obj equal}
	\end{equation}

	The remaining task is to demonstrate the existence of an assignment that assigns $a^t_{n,s}$ to $\left\{a^t_{n,v}|v \in \mathcal{V}_s \right\}$ for all $s \in \mathcal{S}_n$ under constraints formulated in \eqref{eq:obj 2b}-\eqref{eq:obj 2d} and \eqref{eq:two obj equal}.
	
	For a feasible $\left\{a^t_{n,s}|s \in \mathcal{S}_n \right\}$ satisfying \eqref{eq:obj 3b}-\eqref{eq:obj 3d}, it is easy to check that there exists a collection of $a^t_{n,v}$ that satisfies \eqref{eq:obj 2c}-\eqref{eq:obj 2d} and \eqref{eq:two obj equal} simultaneously. This depends on whether there is still a feasible collection of $a^t_{n,v}$ with the additional involvement of \eqref{eq:obj 2b}. Assuming that there is no such collection, we can infer that the solutions of \eqref{eq:obj 2c}-\eqref{eq:obj 2d} and \eqref{eq:two obj equal} in terms of $\left\{a^t_{n,v}|v \in \mathcal{V}_s, s \in \mathcal{S}_n \right\}$ must be or can be turned into a solution where at least one $a^t_{n,v}$ exceeds $l^t_{v,n}$ while others reach their upper bound in \eqref{eq:obj 2b}. This leads to a violation of \eqref{eq:obj 3b} in terms of $\left\{a^t_{n,s}|s \in \mathcal{S}_n \right\}$, which contradicts with the prerequisite that $\left\{a^t_{n,s}|s \in \mathcal{S}_n \right\}$ is feasible.
	
	Therefore, there must have a collection of $a^t_{n,v}$ that satisfies \eqref{eq:obj 2b}-\eqref{eq:obj 2d} and \eqref{eq:two obj equal}. That is to say, each feasible solution of (P3) can be converted to a feasible solution of (P2) without compromising the objective value. Moreover, the conversion can be done in polynomial time. For example, there is at least one efficient way, just increasing $a^t_{n,v}$ from 0 until it reaches $l^t_{v,n}$ or \eqref{eq:two obj equal} becomes tight.
\end{proof}

\begin{corollary} \label{corollary:2}
	(P3) is NP-Hard
\end{corollary}
\begin{proof}
	Let us assume that (P3) is not NP-Hard, meaning that it can be solved in polynomial time. Combining the conclusion of Lemma~\ref{lemma:1}, we can obtain an optimal solution of (P2) from that of (P3) in polynomial time. However, Theorem~\ref{theorem:1} asserts that (P2) is NP-Hard, which means that it cannot be solved in polynomial time. There happens a contradiction. Thus, the hypothesis does not hold, i.e., (P3) is NP-Hard.
\end{proof}

In fact, the physical significance of (P3) is to carry out slice admission control between slices. On the contrary, (P2) is specific to the quota of each VSP. In the scenario where multiple VSPs apply for the same slice, the ultimate goal of the per-slot problem is to solve (P2), while (P3) further reduces the number of variables in (P2) and can better reflect the heterogeneity of slices.

Based on these analyses, we plan to solve (P2) in two stages. First, we design a heuristic algorithm for (P3) to make inter-slice admission decisions and then map the solution of (P3) to that of (P2) subject to an extra constraint, i.e., \eqref{eq:two obj equal}, using a tailored auction mechanism.

\subsection{Inter-slice Admission Decision} \label{subsec:inter slice admission decision}
We propose a heuristic algorithm called Dominant Resource Revenue Efficiency Driven Prioritized Admission (DRREDPA) to solve (P3). The idea of DRREDPA is to spend the same amount of dominant resources to earn as much revenue as possible while trying to maintain the priority differences of slices.

Before describing the algorithm, it is necessary to explain what a dominant resource is and how to calculate the revenue efficiency of a dominant resource. A dominant resource is a resource type that is related to a slice resource demand vector and the available resource vector of an NSP. Supposing that the available resource vector is denoted as $\boldsymbol{C}^*_n$, then the dominant resource of slice~$s$ is the one exhausted first when slice~$s$ monopolizes $\boldsymbol{C}^*_n$. We use $DR_{s|\boldsymbol{C}^*_n}$ to denote the dominant resource of slice~$s$ given $\boldsymbol{C}^*_n$, \eqref{eq:dr} shows the definition. Dominant Resource Revenue Efficiency (DRRE), denoted as $RE_{s|\boldsymbol{C}^*_n}$, refers to the amount of revenue obtained by spending unit dominant resource for slice~$s$ given $\boldsymbol{C}^*_n$, which is formulated in \eqref{eq:re}.
\begin{equation} \label{eq:dr}
	DR_{s|\boldsymbol{C}^*_n} \triangleq \underset{k \in [1,K]}{\text{argmin}} \frac{C^{*,k}_n}{c^k_{n,s}},
\end{equation}
\begin{equation} \label{eq:re}
	RE_{s|\boldsymbol{C}^*_n} = \left. \frac{p_{n,s}}{c^{k^*}_{n,s}} \right|_{k^* = DR_{s|\boldsymbol{C}^*_n}},
\end{equation}

DRREDPA takes the number of active slices, slice demand, and resource capacity as inputs, and outputs the access amount of each slice, \algref{alg:DRREDPA} shows the details. In each time slot~$t$, the inter-slice admission decisions $a^t_{n,s}$ are initiated to zero, DRREDPA first calculates the available resource $\boldsymbol{C}^*_n$ with \eqref{eq:ava} when there are already $A^t_{n,s}$ active slice instances and temporarily $a^t_{n,s}$ to be admitted for $s \in \mathcal{S}_n$.
\begin{equation} \label{eq:ava}
	\boldsymbol{C}^*_n = \boldsymbol{C}_n - \sum_{s \in \mathcal{S}_n}{\boldsymbol{c}_{n,s} \cdot \left(A^t_{n,s}+a^t_{n,s}\right)},
\end{equation}

Afterward, \algref{alg:DRREDPA} checks whether the priority condition in \eqref{eq:obj 3e} is satisfied. If so, it sorts all slices in $\mathcal{S}_n$ in descending order of their DRRE, then grants one extra quota to the first slice whose resource demand does not exceed the available resource and temporary admission quantity is less than the requested, at the premise that the priority condition still holds after doing so. When there have no such slices, DRREDPA returns the current $\left\{a^t_{n,s}|s \in \mathcal{S}_n \right\}$ as the final inter-slice decisions. Otherwise, it goes back to recalculate the available resources $\boldsymbol{C}^*_n$ with updated $\left\{a^t_{n,s}|s \in \mathcal{S}_n \right\}$, and repeats the steps since then.

\begin{algorithm}[!tb]
	\caption{DRREDPA} \label{alg:DRREDPA}
	\SetKwInOut{Input}{Input}
	\SetKwInOut{Output}{Output}
	\Input{1) Active slices $\left\{A^t_{n,s}|s \in \mathcal{S}_n \right\}$;\\
		   2) Slice demands $\left\{l^t_{v,n}|v \in \mathcal{V}_s, s \in \mathcal{S}_n \right\}$;\\
		   3) Resource capacity $\boldsymbol{C}_n$.}
	\Output{Inter-slice decisions $\left\{a^t_{n,s}|s \in \mathcal{S}_n \right\}$.}
	Initiate $\left\{a^t_{n,s}|s \in \mathcal{S}_n \right\}$ to zero, set $flag$ to $True$.
	
	\While{$flag$ is $True$}{
		Set $flag$ to $False$.
		
		Calculate the available resource $\boldsymbol{C}^*_n$ given $\left\{A^t_{n,s}|s \in \mathcal{S}_n \right\}$, $\left\{a^t_{n,s}|s \in \mathcal{S}_n \right\}$ and $\boldsymbol{C}_n$.
		
		\eIf{the priority condition is satisfied given $\left\{a^t_{n,s}|s \in \mathcal{S}_n \right\}$}{
			Calculate $RE_{s|\boldsymbol{C}^*_n}$ for each $s$ in $\mathcal{S}_n$. \label{line:DRREDPA compute RE for Sn}
			
			Sort $s$ in $\mathcal{S}_n$ in descending order of $RE_{s|\boldsymbol{C}^*_n}$.
			
			\For{each $s$ in $\mathcal{S}_n$ \label{line:DRREDPA traverse Sn}}{
				\If{$\boldsymbol{c}_{n,s} \leq \boldsymbol{C}^*_n, a^t_{n,s} < \sum_{v \in \mathcal{V}_s}{l^t_{v,n}}$ \rm{\textbf{and}} the priority condition still holds when $a^t_{n,s}$ increases 1 \label{line:DRREDPA check resource priority}}{
					$a^t_{n,s} \leftarrow a^t_{n,s} + 1$.
					
					Set $flag$ to $True$, \textbf{break}. \label{line:DRREDPA traverse Sn break}
				}
			}
		}{
			Form a new set $\mathcal{S}^*_n$ with the slices in $\mathcal{S}_n$ that violate the priority condition.
			
			Calculate $RE_{s|\boldsymbol{C}^*_n}$ for each $s$ in $\mathcal{S}^*_n$. \label{line:DRREDPA compute RE for Sn star}
			
			Sort $s$ in $\mathcal{S}^*_n$ in descending order of $RE_{s|\boldsymbol{C}^*_n}$.
			
			\For{each $s$ in $\mathcal{S}^*_n$}{
				\If{$\boldsymbol{c}_{n,s} \leq \boldsymbol{C}^*_n, a^t_{n,s} < \sum_{v \in \mathcal{V}_s}{l^t_{v,n}}$}{
					$a^t_{n,s} \leftarrow a^t_{n,s} + 1$.
					
					Set $flag$ to $True$, \textbf{break}. \label{line:DRREDPA traverse Sn star break}
				}
			}
		}
	}
	\textbf{return} $\left\{a^t_{n,s}|s \in \mathcal{S}_n \right\}$.
\end{algorithm}

When \algref{alg:DRREDPA} captures priority violations above, it filters out slices that break the priority condition into set $\mathcal{S}^*_n$, calculates their DRRE, and sorts them in descending order. Then, DRREDPA additionally admits the first slice with residual requests and resource demands that do not exceed the available resources. Similarly, if there does exist such a slice, DRREDPA goes back to recalculate the available resources $\boldsymbol{C}^*_n$ with updated $\left\{a^t_{n,s}|s \in \mathcal{S}_n \right\}$, and repeats the steps since then. Otherwise, it returns the current $\left\{a^t_{n,s}|s \in \mathcal{S}_n \right\}$ as the final inter-slice decisions.

\subsection{Intra-slice Quota Auction} \label{subsec:intra-slice quota auction}
DRREDPA introduced in Section~\ref{subsec:inter slice admission decision} produces a feasible solution $\left\{a^t_{n,s}|s \in \mathcal{S}_n\right\}$ of (P3), NSP~$n$ needs an allocation rule to further distribute quotas $a^t_{n,s}$ to alternative VSPs in $\mathcal{V}_s$ for slices in $\mathcal{S}_n$, so as to get a solution of (P2) with equal objective value as (P3). There may be multiple allocation schemes that can make it, but our expectation is to pursue the one that can balance the allocation of quotas among alternative VSPs within each slice.

Except for \eqref{eq:obj 2b}-\eqref{eq:obj 2d} and \eqref{eq:two obj equal}, we put another requirement on the allocation rule; it should produce a result that achieves the Value-weighted Proportional Fairness (VWPF) among alternative VSPs, whose valuation $b^*_v$ is bigger than the base price $p_{n,s}$. In other words, if $\left\{\hat{a}^t_{n,v}|v \in \mathcal{V}_s\right\}$ is an allocation result of NSP~$n$ that satisfies VWPF within slice~$s$, then for any other feasible solution $\left\{a^t_{n,v}|v \in \mathcal{V }_s\right\}$, \eqref{eq:VWPF} holds:
\begin{equation} \label{eq:VWPF}
	\sum_{v \in \mathcal{V}_s}{\mathbb{I}_{p_{n,s} \leq b^*_v} \cdot b^*_v \cdot \frac{a^t_{n,v} - \hat{a}^t_{n,v}}{\hat{a}^t_{n,v}+\epsilon}} \leq 0,
\end{equation}

\noindent The valuations $\left\{b^*_v|v \in \mathcal{V}_s \right\}$ work as weights in \eqref{eq:VWPF} and $\mathbb{I}_{p_{n,s} \leq b^*_v}$ is an indicator function that takes 1 when the condition is true or 0 otherwise, $\epsilon$ is a positive parameter to avoid zero denominators.

The literature about proportional fairness ensures that the solution of \eqref{eq:VWPF} must maximize the tailored utility function formulated in \eqref{eq:obj 4}. Therefore, the ideal allocation rule must solve the optimization problem described in (P4) for any $s \in \mathcal{S}_n$. Note that \eqref{eq:obj 4a} is introduced to ensure that NSP~$n$ only serves VSPs with bigger valuations than $p_{n,s}$. It works as an auxiliary constraint to make \eqref{eq:obj 4} equivalent to \eqref{eq:VWPF}. \eqref{eq:obj 4b}-\eqref{eq:obj 4e} are directly transferred from \eqref{eq:obj 2b}-\eqref{eq:obj 2d} and \eqref{eq:two obj equal}.
\begin{align}
	\text{(P4): } \underset{a^t_{n,v}}{\max} \sum_{v \in \mathcal{V}_s}{b^*_v \cdot \text{log}\left(a^t_{n,v} + \epsilon \right)} \label{eq:obj 4} \\
	\text{s.t. } a^t_{n,v} \cdot p_{n,s} \leq a^t_{n,v} \cdot b^*_v, \forall v \in \mathcal{V}_s, \tag{\ref{eq:obj 4}a} \label{eq:obj 4a} \\
	a^t_{n,v} \leq l^t_{v,n}, \forall v \in \mathcal{V}_s, \tag{\ref{eq:obj 4}b} \label{eq:obj 4b} \\
	a^t_{n,v} \geq 0, \forall v \in \mathcal{V}_s, \tag{\ref{eq:obj 4}c} \label{eq:obj 4c} \\
	a^t_{n,v} \in \mathbb{Z}, \forall v \in \mathcal{V}_s, \tag{\ref{eq:obj 4}d} \label{eq:obj 4d} \\
	\sum_{v \in \mathcal{V}_s}{a^t_{n,v}} = a^t_{n,s}, \tag{\ref{eq:obj 4}e} \label{eq:obj 4e}
\end{align}

Let us suppose that $\left\{b^*_v|v \in \mathcal{V}_n \right\}$ is honestly revealed to NSP~$n$, we propose an allocation rule called Larger Increment First~(LIF) to find the optimal assignment of (P4). LIF filters out VSPs in $\mathcal{V}_s$ whose valuation is lower than $p_{n,s}$ in advance and constructs a set of increments for each VSP remained, the increments correspond to the utility steps yielded to \eqref{eq:obj 4} when raising $a^t_{n,v}$ from zero to $l^t_{v,n}$ with step length 1. Then, LIF sorts these increments in descending order, only the first $a^t_{n,s}$ increments are preserved. LIF finally allocates one quota to each selected increment, and the VSP to which the increment belongs will get the quota attached.

\begin{lemma} \label{lemma:2}
	LIF gives the optimal solution of (P4).
\end{lemma}
\begin{proof}
	We can easily verify that the operation of constructing and selecting increments strictly guarantees \eqref{eq:obj 4a}-\eqref{eq:obj 4e}, under the assumption that VSPs who eager for slice~$s$ are able and willing to pay the base price $p_{n,s}$. After all the increments are sorted, each quota is assigned to the larger one first, so the final utility is the largest among all feasible allocation results.
\end{proof}

Nevertheless, the values $\left\{b^*_v | v \in \mathcal{V}_s \right\}$ are private, and NSP~$n$ needs additional tricks to obtain this information. Fortunately, an auction mechanism with the property of Dominant Strategy Incentive Compatible (DSIC) can achieve it. Therefore, we manage to devise such an auction mechanism to implement the intra-slice quota allocation. The parities involved are defined as follows:
\begin{itemize}
	\item \textbf{Auctioneer}. NSP~$n$ is the auctioneer. Each supported slice~$s$ corresponds to an auction in each time slot~$t$. The items sold are the admission quotas $a^t_{n,s}$ of the slice and the social welfare is given by \eqref{eq:obj 4}.
	
	\item \textbf{Bidder}. VSPs that eager for slice~$s$ are bidders. Their valuations $\left\{b^*_v | v \in \mathcal{V}_s \right\}$ are the incomes that they can earn per time slot after obtaining a quota. Although we assume that AVSPs sell the service at the same price, their valuations may differ due to differences in Operating Expenditure.
\end{itemize}

An auction mechanism includes an allocation rule, which we have already proposed (i.e., LIF), and a charging rule. It is the charging rule that forces VSPs to bid truthfully. The \textit{Myerson's Lemma} \cite{roughgarden2016twenty} says that as long as the allocation rule is monotonically non-decreasing to the bid, there must be a unique charging rule that makes the auction mechanism DSIC and it is quite elegant; just charge the prices at which the allocation curve steps.

\begin{lemma} \label{lemma:3}
	LIF is monotonically non-decreasing to the bid. In other words, if one VSP raises its bid while others stay the same, it gets at least as many quotas as with a smaller bid.
\end{lemma}
\begin{proof}
	Raising the bid of one VSP will only promote its increments without affecting those of others. The rankings of the new increments are only likely to move forward, so the VSP will not experience a loss of quotas.
\end{proof}

\begin{corollary} \label{corollary:unique charging rule}
	There exists a unique charging rule that makes LIF DSIC.
\end{corollary}
\begin{proof}
	Combining \textit{Myerson's Lemma} and Lemma~\ref{lemma:3}, we can infer that LIF meets the requirements of \textit{Myerson's Lemma} for the allocation rule. Therefore, there must be a unique charging rule, which together with LIF constitutes an auction mechanism with the DSIC property.
\end{proof}

With \corolref{corollary:unique charging rule}, we next elaborate on this unique charging rule in \propref{proposition:CPRP}, which is named Critical Pricing with Reserved Prices (CPRP).

\begin{proposition} \label{proposition:CPRP}
	(CPRP) Taking VSP~$v$ belonging to $\mathcal{V}_s$ as an example. The charging rule that makes LIF DSIC just prices the quota attached to the $i$-th smallest preserved increment of VSP~$v$ according to \eqref{eq:critical price}:
	\begin{equation} \label{eq:critical price}
		p^{-i}_{n,v} = \max \left\{b_v \cdot \frac{\delta^{i}_{\mathcal{V}_s \backslash v}}{\Delta^{-i}_{n,v}},\; p_{n,s}\right\},\, i \in \left\{1, \dots, a^t_{n,v} \right\},
	\end{equation}
	
	\noindent where $b_v$ is the valuation reported by VSP~$v$, $\Delta^{-i}_{n,v}$ is the $i$-th smallest increment of VSP~$v$ that is preserved by LIF. $\delta^i_{\mathcal{V}_s \backslash v}$ is the $i$-th highest loser of increments belonging to $\mathcal{V}_s \backslash v$, if there is no such loser, it takes 0.
\end{proposition}
\begin{proof}
	We can easily find that $p^{-i}_{n,v}$ is never lower than the base price $p_{n,s}$. Since $\delta^{i}_{\mathcal{V}_s \backslash v}$ is less than $\Delta^{-i}_v$, $p^{-i}_{n,v}$ is guaranteed less than $b_v$. This means that the price given by \eqref{eq:critical price} will be limited to $\left[p_{n,s}, b_v \right]$. As long as the valuation of VSP~$v$ is higher than the base price, which is an aforementioned premise, this charging rule will not cause VSP~$v$ to generate negative returns.
	
	Let us concentrate on the left formula inside the curly braces of \eqref{eq:critical price}, we name it Critical Price (CP) for convenient. There may be two cases:
	\begin{enumerate}
		\item \textit{CP is bigger than $p_{n,s}$ and $p^{-i}_{n,v}$ is set to CP}. In this case, if VSP~$v$ lowers the bid to be slightly smaller than CP, then its $i$-th smallest increment preserved by LIF will be defeated by $\delta^{i}_{\mathcal{V}_s \backslash v}$. In other words, the allocation curve of VSP~$v$ produces a unit step at CP. Therefore, the quota attached to the $i$-th smallest increment is indeed priced at the step point.
		
		\item \textit{CP is smaller than $p_{n,s}$ and $p^{-i}_{n,v}$ is set to $p_{n,s}$}. Once VSP~$v$ shrinks the bid to be slightly smaller than $p_{n,s}$, it will lose the $i$-th last quota along with those preceding. That is because $p_{n,s}$ is the reserved price identified by LIF. Therefore, quotas are still priced at the step point in this case.
	\end{enumerate}
	
	In conclusion, CPRP always charges quotas at the step points of the allocation curve yielded by LIF. This is exactly what \textit{Myerson's Lemma} requires for the charging rule.
\end{proof}

We have demonstrated in \propref{proposition:CPRP} that CPRP strictly abides by the design principles of \textit{Myerson's Lemma} on charging rules. Therefore, LIF and CPRP cooperate to realize an auction mechanism with reserved price and DSIC characteristic. We name the tailored mechanism Value-weighted Proportional Fairness Auction (VWPFA) and summarize it in \algref{alg:VWPFA}.

\begin{algorithm}[!tb]
	\caption{VWPFA} \label{alg:VWPFA}
	\SetKwInOut{Input}{Input}
	\SetKwInOut{Output}{Output}
	\Input{1) Inter-slice decisions $\left\{a^t_{n,s}|s \in \mathcal{S}_n \right\}$;\\
		2) Slice demands $\left\{l^t_{v,n}|v \in \mathcal{V}_s, s \in \mathcal{S}_n \right\}$;\\
		3) bids $\left\{b_v| v \in \mathcal{V}_s, s \in \mathcal{S}_n\right\}$.}
	\Output{Intra-slice decisions $\left\{a^t_{n,v}|v \in \mathcal{V}_s, s \in \mathcal{S}_n \right\}$ and prices $\left\{ p^{-i}_{n,v}|i \in \left\{1, \dots, a^t_{n,v} \right\}, v \in \mathcal{V}_s, s \in \mathcal{S}_n \right\}$.}
	\For{each $s$ in $\mathcal{S}_n$ \label{line:VWPFA traverse Sn}}{
			\tcp{Largest Increment First (LIF)}
			Filter out VSPs whose bid $b_v$ is less than $p_{n,s}$ from $\mathcal{V}_s$. \label{line:VWPFA filter VSP}
			
			\For{each remaining $v$ in $\mathcal{V}_s$ \label{line:VWPFA traverse Vs for increments}}{
				Compute the increments of \eqref{eq:obj 4} when raising $a^t_{n,v}$ from zero to $l^t_{v,n}$ with step length 1. \label{line:VWPFA compute increments}
			}
			Sort all increments in descending order and assign one quota to each of the top $a^t_{n,s}$ ones. \label{line:VWPFA sort increments}
			
			\For{each remaining $v$ in $\mathcal{V}_s$ \label{line:VWPFA traverse Vs for allocation}}{
				$a^t_{n,v} \leftarrow$ the number of quotas attached to increments belonging to VSP~$v$. \label{line:VWPFA assign quota}
				
				\tcp{Critical Pricing with Reserved Prices (CPRP)}
				\For{each $i \in \left\{1, \dots, a^t_{n,v} \right\}$ \label{line:VWPFA price all quota for v}}{
					Find the $i$-th higest loser $\delta^{i}_{\mathcal{V}_s \backslash v}$ and the $i$-th smallest winner $\Delta^{-i}_{v}$. \label{line:VWPFA find competitor}
					
					Compute the price $p^{-i}_{n,v}$ with \eqref{eq:critical price}.
					
					
					\eIf{$p^{-i}_{n,v} = p_{n,s}$}{
							Price the unpriced quotas granted to VSP~$v$ with $p_{n,s}$, \textbf{break}. \label{line:VWPFA bulk pricing}
						}{
							Price the $i$-th last quota of VSP~$v$ with $p^{-i}_{n,v}$. \label{line:VWPFA price the last quota}
						}
				}
			}
			Assign the residual quotas (when there have) to the VSPs skipped in \lineref{line:VWPFA filter VSP} at the price $p_{n,s}$. \label{line:VWPFA penalty}
	}
	\textbf{return} \newline $\left\{a^t_{n,v}|v \in \mathcal{V}_s, s \in \mathcal{S}_n \right\}$, $\left\{ p^{-i}_{n,v}|i \in \left\{1, \dots, a^t_{n,v} \right\}, v \in \mathcal{V}_s, s \in \mathcal{S}_n \right\}$.
\end{algorithm}

It is worth mentioning that \lineref{line:VWPFA penalty} in \algref{alg:VWPFA} is to deal with the situation where some VSPs consider that the value of slice~$s$ does not worth $p_{n,s}$, but still initiate slices requests. They are deferred until other VSPs have obtained their quotas and are then allocated the possible remaining quotas at the well-known base price $p_{n,s}$ obligatorily. It acts as a supplement to prevent those malicious tenants who do not agree with the value of slice from initiating requests casually.

With the assistance of VWPFA, the only dominant strategy for VSPs is to bid honestly. In this way, NSP can prevent strategic VSPs from seeking more benefit by inflating or under-quoting, so as to achieve the value-weighted proportional fairness within the slice. Moreover, the introduction of the reserved price ensures that the real income of NSP is not lower than the base revenue defined in \eqref{eq:obj 3}. In addition, VWPFA is also compatible with slices subscribed by a single VSP. In this case, there are no competitors and no $\delta^{i}_{\mathcal{V}_s \backslash v}$ can be found in \lineref{line:VWPFA find competitor}, so all quotas will be exclusively assigned to the only VSP (in \lineref{line:VWPFA assign quota}) at the price $p_{n,s}$ (in \lineref{line:VWPFA bulk pricing}). It is exactly what competition-free slices desire.

\begin{algorithm}[!tb]
	\caption{MPSAC} \label{alg:MPSAC}
	\SetKwInOut{Input}{Input}
	\SetKwInOut{Output}{Output}
	\Input{1) Active slices $\left\{A^t_{n,s}|s \in \mathcal{S}_n \right\}$;\\
		   2) Slice demands $\left\{l^t_{v,n}|v \in \mathcal{V}_s, s \in \mathcal{S}_n \right\}$;\\
		   3) Resource capacity $\boldsymbol{C}_n$.}
	\Output{Intra-slice decisions $\left\{a^t_{n,v}|v \in \mathcal{V}_s, s \in \mathcal{S}_n \right\}$ and prices $\left\{ p^{-i}_{n,v}|i \in \left\{1, \dots, a^t_{n,v} \right\}, v \in \mathcal{V}_s, s \in \mathcal{S}_n \right\}$.}
	
	\tcp{The inter-slice decision}
	Invoke \textit{DRREDPA} in \algref{alg:DRREDPA} given $\left\{A^t_{n,s}|s \in \mathcal{S}_n \right\}$, $\left\{l^t_{v,n}|v \in \mathcal{V}_s, s \in \mathcal{S}_n \right\}$ and $\boldsymbol{C}_n$ as inputs to get $\left\{ a^t_{n,s}|s \in \mathcal{S}_s \right\}$.
	
	\tcp{The intra-slice auction}
	Let VSPs report their valuations to get $\left\{b_v| v \in \mathcal{V}_s, s \in \mathcal{S}_n\right\}$.
	
	Invoke \textit{VWPFA} in \algref{alg:VWPFA} given $\left\{ a^t_{n,s}|s \in \mathcal{S}_s \right\}$, $\left\{l^t_{v,n}|v \in \mathcal{V}_s, s \in \mathcal{S}_n \right\}$ and $\left\{b_v| v \in \mathcal{V}_s, s \in \mathcal{S}_n\right\}$ as inputs to get $\left\{a^t_{n,v}|v \in \mathcal{V}_s, s \in \mathcal{S}_n \right\}$ and $\left\{ p^{-i}_{n,v}|i \in \left\{1, \dots, a^t_{n,v} \right\}, v \in \mathcal{V}_s, s \in \mathcal{S}_n \right\}$. 
	
	\textbf{return} \newline $\left\{a^t_{n,v}|v \in \mathcal{V}_s, s \in \mathcal{S}_n \right\}$, $\left\{ p^{-i}_{n,v}|i \in \left\{1, \dots, a^t_{n,v} \right\}, v \in \mathcal{V}_s, s \in \mathcal{S}_n \right\}$.
\end{algorithm}

Combining the algorithms designed for (P3) and (P4), we propose the Multi-participant Slice Admission Control (MPSAC) algorithm depicted in \algref{alg:MPSAC}. It takes active slice information, slice requests and resource capacity of the executor (i.e., the NSP) as input, while outputs the admission quotas for each contracted VSP together with the price of each quota. MPSAC is triggered at each time slot and firstly invokes DRREDPA (i.e., \algref{alg:DRREDPA}) fed with the input to decide the aggregate quotas delivered to each slice. Afterwards, it collects the valuation information from VSPs and passes it to VWPFA along with the slice quotas obtained before to get the final decision. MPSAC gives a polynomial-time method to deal with the per-slot optimization problem (P2) and works as an indirect approach to pursue a sub-optimal solution of (P1) in a time-efficient manner.

\subsection{Time Complexity Analysis}

In order to get the complexity of MPSAC, we demonstrate the complexity of DRREDPA and VWPFA respectively.

We first analyze \algref{alg:DRREDPA}. The initiation operation needs to set $|\mathcal{S}_n|$ variables, so it takes $O(|\mathcal{S}_n|)$. Inside the \textit{while loop}, calculating the available resources requires $|\mathcal{S}_n|$ matrix (in dimension $K$) multiplications and $|\mathcal{S}_n|$ matrix additions, it takes $O(K \cdot |\mathcal{S}_n|)$. Calculating revenue efficiency for one slice requires $K+1$ divisions and $K-1$ comparisons, so it takes $O(K \cdot |\mathcal{S}_n|)$ for all slices. Sorting takes $O( |\mathcal{S}_n| \cdot \text{log}(|\mathcal{S}_n|) )$. The condition judgment in \lineref{line:DRREDPA check resource priority} needs $K+|\mathcal{S}_n|$ comparisons, 2 additions and 1 division, so the \textit{for loop} in \lineref{line:DRREDPA traverse Sn} takes $O(K \cdot |\mathcal{S}_n| + |\mathcal{S}_n|^2)$. We conclude that \lineref{line:DRREDPA compute RE for Sn} to \lineref{line:DRREDPA traverse Sn break} takes $O(2K \cdot |\mathcal{S}_n| + |\mathcal{S}_n| \cdot \text{log}(|\mathcal{S}_n|) + |\mathcal{S}_n|^2)$. A similar analysis applies to \lineref{line:DRREDPA compute RE for Sn star} to \lineref{line:DRREDPA traverse Sn star break}, it takes at most $O(2K \cdot |\mathcal{S}_n| + |\mathcal{S}_n| \cdot \text{log}(|\mathcal{S}_n|))$ as $|\mathcal{S}^*_n| < |\mathcal{S}_n|$. Therefore, executing a \textit{while loop} takes up to $O(3K \cdot |\mathcal{S}_n| + |\mathcal{S}_n| \cdot \text{log}(|\mathcal{S}_n|) + |\mathcal{S}_n|^2)$. Limited by the capacity of physical resources, the maximum number of times that the \textit{while loop} can be executed is equal to the maximum number of slices that NSP~$n$ can accommodate, which is a constant that is independent from $K$ and $|\mathcal{S}_n|$. We omit the constants in complexity expression, then the complexity of DRREDPA is $O(K \cdot |\mathcal{S}_n| + |\mathcal{S}_n| \cdot \text{log}(|\mathcal{S}_n|) + |\mathcal{S}_n|^2)$.

Now we consider \algref{alg:VWPFA}. Inside the \textit{for loop} in \lineref{line:VWPFA traverse Sn}, the operation of traversing VSP takes $O(|\mathcal{V}_s|)$. A total of $\sum_{v \in \mathcal{V}_s}{l^t_{v,n}}$ subtractions are required to calculate the increments in \lineref{line:VWPFA compute increments}. \lineref{line:VWPFA sort increments} sorts these increments using quicksort. Therefore, \lineref{line:VWPFA compute increments} and \lineref{line:VWPFA sort increments} take a total of $O\left(\left[1 + \text{log}(\sum_{v \in \mathcal{V}_s}{l^t_{v,n}}) \right] \cdot \sum_{v \in \mathcal{V}_s}{l^t_{v,n}} \right)$. Allocating quotas for all $v$ in $\mathcal{V}_s$ requires at most $a^t_{n,s}$ additions, \lineref{line:VWPFA assign quota} takes $O(a^t_{n,s})$ totally. Pricing a quota for VSP~$v$ requires at most 1 division, 1 multiplication, and 1 comparison in \lineref{line:VWPFA find competitor} to \lineref{line:VWPFA price the last quota}, so it takes up to $O(3a^t_{n,v})$ to determine all the prices for quotas awarded to VSP~$v$ with the \textit{for loop} in \lineref{line:VWPFA price all quota for v}, then it takes $O\left(3 \cdot \sum_{v \in \mathcal{V}_s}{a^t_{n,v}}\right)$ to determine those for all VSPs in $\mathcal{V}_s$, which can be equivalently written as $O(3a^t_{n,s})$. Therefore, executing the \textit{for loop} in \lineref{line:VWPFA traverse Sn} for one time takes up to $O\big(|\mathcal{V}_s| + \left[1 + \text{log}(\sum_{v \in \mathcal{V}_s}{l^t_{v,n}}) \right] \cdot \sum_{v \in \mathcal{V}_s}{l^t_{v,n}} + 3a^t_{n,s} \big)$, executing the entire \textit{for loop} takes $O\big(|\mathcal{V}_n| + \sum_{s \in \mathcal{S}_n}{\left[1 + \text{log}(\sum_{v \in \mathcal{V}_s}{l^t_{v,n}}) \right] \cdot \sum_{v \in \mathcal{V}_s}{l^t_{v,n}}} +3 \cdot \sum_{s \in \mathcal{S}_n}{a^t_{n,s}} \big)$. Since $\sum_{s \in \mathcal{S}_n}{a^t_{n,s}}$ cannot exceed the maximum capacity, it can be regarded as a constant with a definite upper bound. Neglecting low-order terms and constants, the complexity of \algref{alg:VWPFA} is $O\big(|\mathcal{V}_n| + \sum_{s \in \mathcal{S}_n}{\left[ \text{log}(\sum_{v \in \mathcal{V}_s}{l^t_{v,n}}) \cdot \sum_{v \in \mathcal{V}_s}{l^t_{v,n}} \right]} \big)$.

With the above analysis, we assert that MPSAC is a polynomial-time algorithm with a complexity of $O\big( K \cdot |\mathcal{S}_n| + |\mathcal{S}_n| \cdot \text{log}(|\mathcal{S}_n|) + |\mathcal{S}_n|^2 + |\mathcal{V}_n| + \sum_{s \in \mathcal{S}_n}{\left[ \text{log}(\sum_{v \in \mathcal{V}_s}{l^t_{v,n}}) \cdot \sum_{v \in \mathcal{V}_s}{l^t_{v,n}} \right]} \big)$.

\section{Numerical Simulations} \label{sec:numerical simulations}

In this section, we describe the evaluation of MPSAC (i.e., DRREDPA together with VWPFA). First, the algorithms adopted for comparison are introduced briefly. Then, we give the metrics used for assessing their performance, followed by the details of the simulation environment and parameter settings. Finally, we discuss the results obtained.

\subsection{Comparison Algorithms}

Three existing algorithms designed for inter-slice admission control are employed as comparisons. Their main ideas are described below, with minor adjustments when necessary to make them applicable to our problem.

\begin{itemize}
	\item \textbf{PAGE} \cite{challa2019network}. It broadly divides services into two kinds: high-priority and low-priority, then statically and exclusively allocates a certain percentage of total resources to each service. We simply extend PAGE to the multi-service scenario investigated in this work by using the relative allocation ratio given by the authors.
	
	\item \textbf{MQSAC} \cite{han2020multiservice}. It characterizes the admission strategy with a preference matrix, each column of the matrix corresponds to the order in which slice admissions are attempted for a given state. MQSAC randomly generates multiple preference matrices and tests the performance when using different preference matrices. Herein, we take the average performance under all preference matrices as its benchmark.
	
	\item \textbf{DSARA} \cite{villota2022admission}. It is a DQN-based algorithm, where the state is defined as a vector of available resources in percentage, and the action is a vector of integers indicating the weights of slices that an admission result should obey.
\end{itemize}

As for intra-slice requests arriving at the same time, few works consider performing more complex admission operations between competing VSPs. Therefore, we adopt a typical baseline that allocates quotas on the proportion (OP) of each VSP's demand \cite{caballero2019network}.

We combine the above inter-slice algorithms and the intra-slice algorithm OP to construct several complete slice admission control strategies as the comparisons of MPSAC, which consists of DRREDPA and VWPFA.

\subsection{Performance Metrics}

In this work, we mainly focus on the following five indicators:

\begin{enumerate}
	\item \textbf{Average base revenue}. This metric corresponds to the original optimization objective formulated in \eqref{eq:obj 1}.
	
	\item \textbf{Inter-slice fairness}. This metric is formulated in \eqref{eq:inter slice fairness}, it visualizes whether the priority condition given by \eqref{eq:obj 1e} is met in each slot, and reflects the degree of fairness between slices.
	
	\item \textbf{Average VWPF}. The value of VWPF given in \eqref{eq:obj 4} is very sensitive to the available quotas $a^t_{n,s}$, while the quotas within a slice at the same time slot~$t$ may be quite different even the same inter-slice algorithm but distinct intra-slice algorithms are used. Therefore, in order to eliminate the confusion as much as possible when presenting the intra-slice proportional fairness, it will be better to perform slot averaging on the value of VWPF, as shown in \eqref{eq:AVWPF}:
	\begin{equation} \label{eq:AVWPF}
		\frac{1}{T} \sum_{t=1}^{T}{\sum_{v \in \mathcal{V}_s}{b^*_v \cdot \text{log}\left(a^t_{n,v} + \epsilon \right)}},
	\end{equation}
	
	\item \textbf{Average actual revenue}. This metric accounts for the slot average of NSP's real revenue, and it helps to check whether VWPFA has further improved the income of a NSP as expected. \eqref{eq:average actual revenue} shows the formula:
	\begin{equation} \label{eq:average actual revenue}
		\frac{1}{T} \sum^{T}_{t=1}{\sum_{s \in \mathcal{S}_n}{\left[\sum_{j=1}^{A^t_{n,s}}{p^{j}_{n,s}} + \sum_{v \in \mathcal{V}_s}{\sum_{i=1}^{a^t_{n,v}}{p^{-i}_{n,v}}}\right]}},
	\end{equation}

	where $p^j_{n,s}$ and $p^{-i}_{n,v}$ are the fees actually charged by NSP~$n$ for a specific slice instance.
	
	\item \textbf{Average time cost}. This metric evaluates the time overhead of admission procedures and takes the mean from multiple repeated simulations to reduce accidental deviations.
\end{enumerate}
\subsection{Simulation Environment}

The simulation is implemented on a server equipped with Intel(R) Xeon(R) Gold 5218R CPU @ 2.10GHz, NVIDIA GeForce RTX 3080 GPU, and 252G memory. The programming environment is Python 3.10.4, and Pytorch 1.13.1 is adopted as the machine learning platform. The driver versions are 515.65.01 and 11.7 for GPU and CUDA respectively.

\subsection{Parameter Settings}

Due to privacy issues, the network operators can hardly reveal a real trace data that contains detailed multidimensional resource overhead. Therefore, we follow the practice of comparison algorithms to use numerical simulations.

The market consists of 2 NSPs and 6 VSPs. A total of 5 slice types are supported, among which NSP~1 provides types $\left\{1, 2, 3, 4\right\}$, NSP~2 provides types $\left\{2, 3, 4, 5\right\}$. The set of tenants (i.e., VSPs) for each kind of slice is set to $\left\{1\right\}$, $\left\{2\right\}$, $\left\{3, 4\right\}$, $\left\{5\right\}$ and $\left\{6\right\}$ respectively. The arrival rate of each kind of slice request is set as a multiple of the base arrival rare $\lambda^G$, where the coefficients are 2, 1.5, 2.5, 1, and 1.5 respectively. By adjusting $\lambda^G$, we simulate different stress conditions to test the performance of algorithms. Tab.~\ref{tab:param setting} shows the values of all important parameters, including the resource capacity of NSPs, the resource demands and the base prices of slices, the rational behavior parameters of subscribers, and the VSPs' real valuations of slices.

\begin{table}[h]
	\renewcommand{\arraystretch}{1.5}
	\begin{center}
		\caption{Parameter Settings.}
		\label{tab:param setting}
		\begin{tabular}{ c c c c }
			\hline
			\textbf{Symbol} & \textbf{Value} & \textbf{Symbol} & \textbf{Value}\\
			\hline
			$\mathcal{S}_1$ & $\left\{1, 2, 3, 4\right\}$ & $\lambda^G$ & $\left\{2, 2.5, 3, 3.5, 4\right\}$ \\
			$\mathcal{S}_2$ & $\left\{2, 3, 4, 5\right\}$ & $\lambda^G_1$ & $2 \lambda^G$ \\
			$\boldsymbol{c}_{1,1}$ & $(0.5, 0.35, 0.35)^T$ & $\lambda^G_2$ & $1.5 \lambda^G$ \\
			$\boldsymbol{c}_{1,2}$ & $(0.7, 0.5, 0.45)^T$ & $\lambda^G_3$ & $2.5 \lambda^G$ \\
			$\boldsymbol{c}_{1,3}$ & $(0.7, 0.65, 0.6)^T$ & $\lambda^G_4$ & $\lambda^G$ \\
			$\boldsymbol{c}_{1,4}$ & $(0.8, 0.8, 0.8)^T$ & $\lambda^G_5$ & $1.5 \lambda^G$ \\
			$\boldsymbol{c}_{2,2}$ & $(0.7, 0.5, 0.45)^T$ & $\lambda^L_1$, $\lambda^W_1$ & 4.0, 4.0 \\
			$\boldsymbol{c}_{2,3}$ & $(0.7, 0.65, 0.6)^T$ & $\lambda^L_2$, $\lambda^W_2$ & 3.0, 4.0 \\
			$\boldsymbol{c}_{2,4}$ & $(0.8, 0.8, 0.8)^T$ & $\lambda^L_3$, $\lambda^W_3$ & 4.0, 5.0 \\
			$\boldsymbol{c}_{2,5}$ & $(0.7, 0.7, 0.9)^T$ & $\lambda^L_4$, $\lambda^W_4$ & 4.0, 3.0 \\
			$\boldsymbol{C}_1$ & $(25, 20, 20)^T$ & $\lambda^L_5$, $\lambda^W_5$ & 3.0, 4.0 \\
			$\boldsymbol{C}_2$ & $(20, 20, 25)^T$ & $\alpha$, $\beta$, $\epsilon$ & 0.5, 0.1, 1.0 \\
			$p_{1,1}$, $p_{1,2}$ & 1.0, 1.4 & $\mathcal{V}_1$, $\mathcal{V}_2$, $\mathcal{V}_3$ & $\left\{1\right\}$, $\left\{2\right\}$, $\left\{3, 4\right\}$ \\
			$p_{1,3}$, $p_{1,4}$ & 1.6, 2.0 & $\mathcal{V}_4$, $\mathcal{V}_5$ & $\left\{5\right\}$, $\left\{6\right\}$ \\
			$p_{2,2}$, $p_{2,3}$ & 1.4, 1.6 & $b^*_1$, $b^*_2$, $b^*_3$ & 2.5, 3.5, 4.5 \\
			$p_{2,4}$, $p_{2,5}$ & 2.0, 2.3 & $b^*_4$, $b^*_5$, $b^*_6$ & 6.0, 5.0, 5.5 \\
			\hline 
		\end{tabular}
	\end{center}
\end{table}

During the simulations, we replace the SAC algorithm of NSP~2 while fixing that of NSP~1 to MPSAC to get rid of irrelevant interference. 50 repeated tests are carried out under each set of parameter settings. By the way, 100 distinct preference matrices are randomly generated for MQSAC.

\subsection{Performance Evaluation} \label{sec:performance evaluation}

In this subsection, we exhibit the above five metrics of NSP~2 when different algorithms are employed.

\begin{figure*}[!ht]
	\centering
	\subfloat[$\lambda^G = 2$]{
		\includegraphics[width=0.33\linewidth]{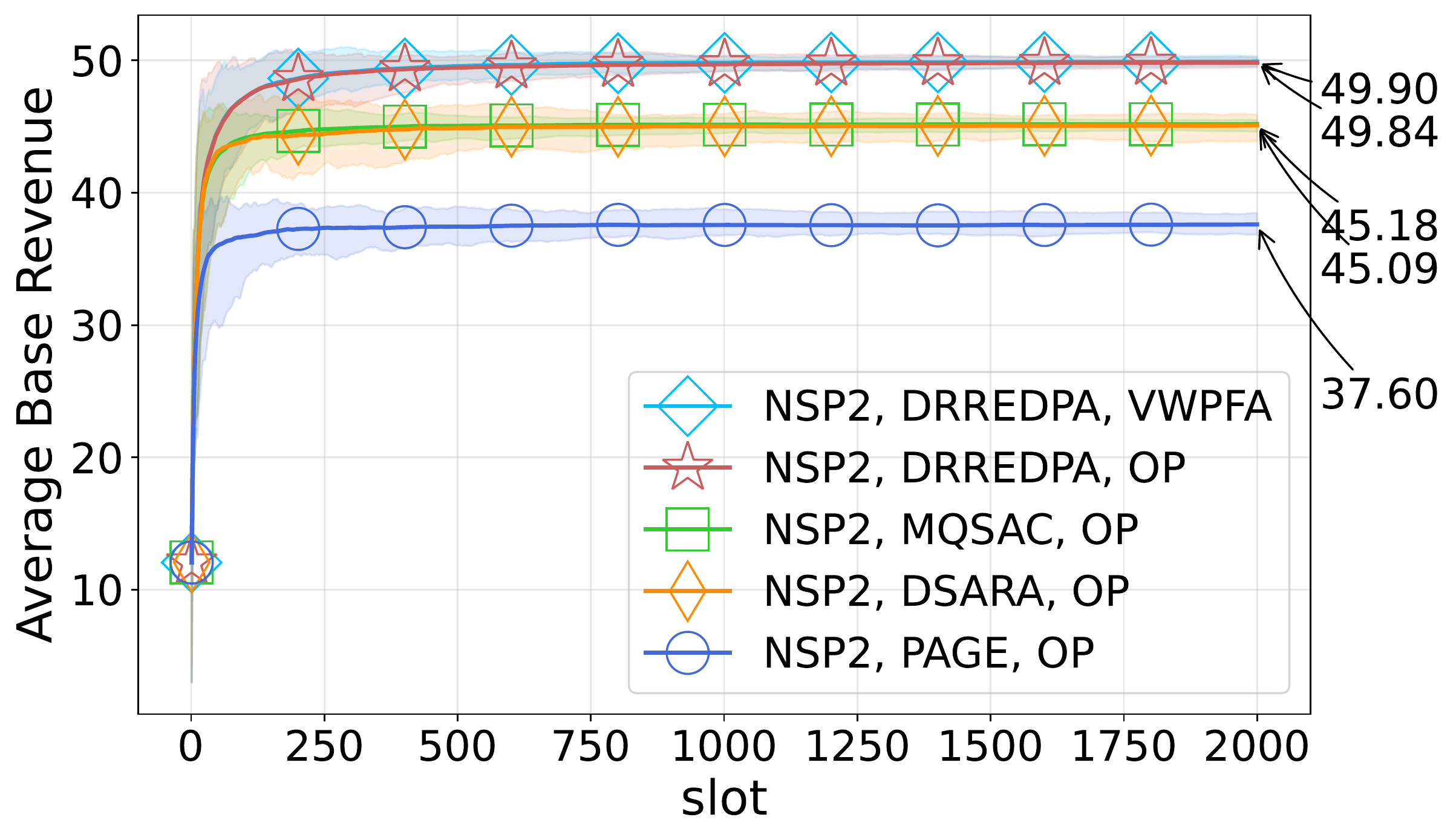}
		\label{fig:long-term average base revenue lambda 2}
	}
	\subfloat[$\lambda^G = 3$]{
		\includegraphics[width=0.33\linewidth]{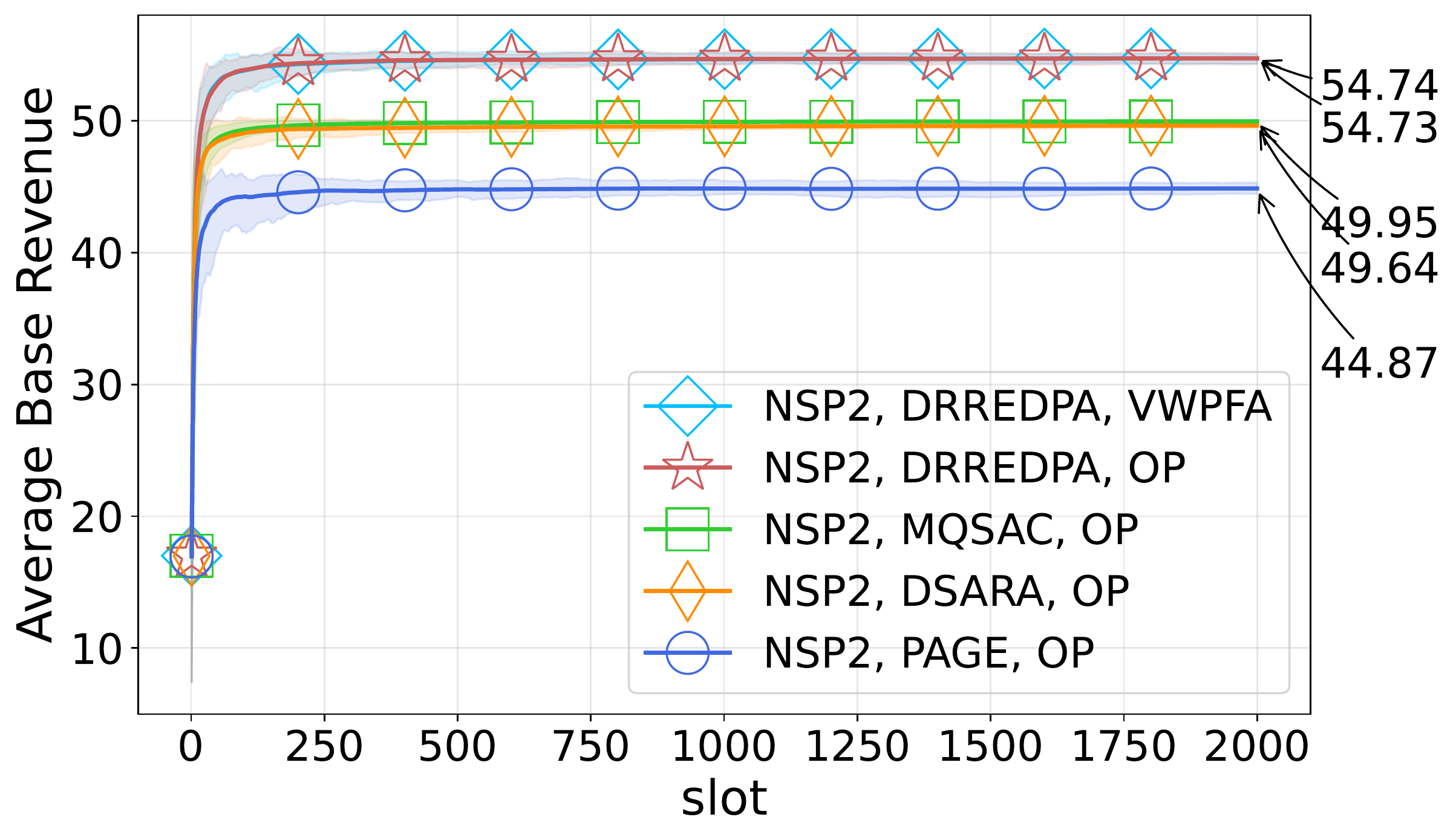}
		\label{fig:long-term average base revenue lambda 3}
	}
	\subfloat[$\lambda^G = 4$]{
		\includegraphics[width=0.33\linewidth]{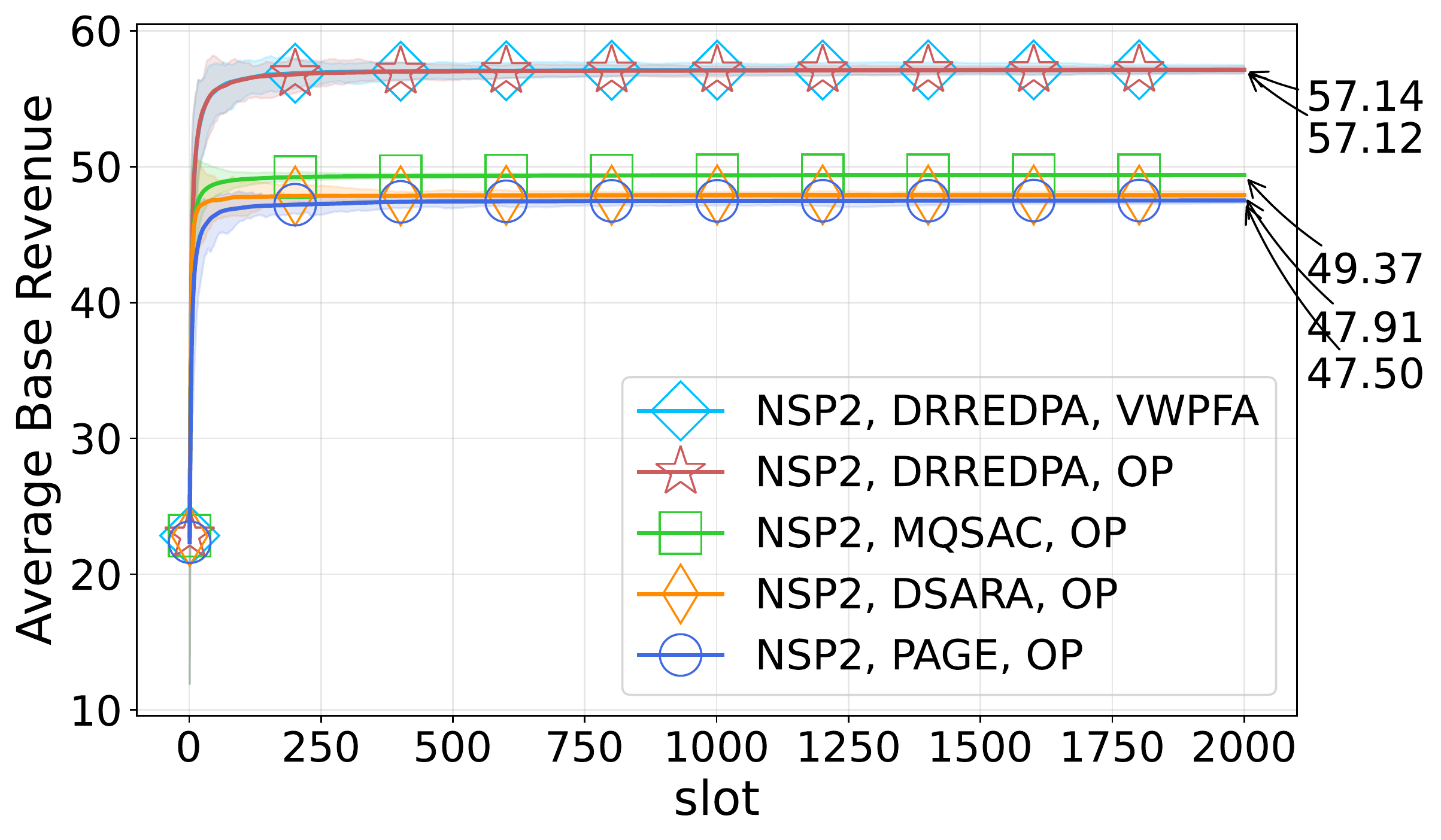}
		\label{fig:long-term average base revenue lambda 4}
	}
	\caption{Long-term average base revenue}
	\label{fig:long-term average base revenue}
\end{figure*}

\subsubsection{Average base revenue}

The long-term average base revenue is depicted in \figref{fig:long-term average base revenue} for 2000 slots, where the solid line shows the mean result of 50 repeated experiments under the corresponding algorithm, and the surrounding shadow is the envelope formed by these parallel data. DRREDPA with VWPFA (i.e., MPSAC), as well as DRREDPA with OP, have always maintained obvious advantages; the lower boundaries of the two are even higher than the upper boundaries of others. MQSAC and DSARA take the second place, and PAGE is the worst one. This phenomenon does not change substantially with the fluctuation of $\lambda^G$, the reasons of which are analyzed below.

MQSAC lacks the ability to design a customized preference matrix, which limits its performance to a certain extent. The more complex the scene is, the more obvious the negative impact yielded by this shortage will be. This is the main reason why MQSAC is inferior to the algorithm we proposed. Although DSARA is a DQN-based algorithm, its actual performance is still not as good as MPSAC. DSARA inherits the typical drawback of Q-learning with limited action space, the optimal action in some states may not be included in the action space, thereby compromising the optimization space of the model. Moreover, due to the influence of rational behavior from multiple participants, the environment observed by individual NSPs does not necessarily satisfy the Markov property, which seriously degrades the performance of reinforcement learning methods. The gap of average base revenue between PAGE and other algorithms shrinks with the growth of  $\lambda^G$. When $\lambda^G$ increases, the total number of requests initiated by subscribers rises, and more high-priority requests appear. In that circumstances, the mechanism of PAGE which reserves more resources for high-priority requests takes effect, leading to a promotion in revenue. However, the under-utilization of resources caused by static allocation cannot be eliminated thoroughly, so it can never catch up with other algorithms.

We can also see from \figref{fig:long-term average base revenue} that the curves and shaded areas of DREEDPA with VWPFA and DREEDPA with OP are highly consistent. It demonstrates that the influence on the long-term average base revenue through changing the intra-slice allocation scheme from OP to VWPFA is completely negligible. This is in line with the analysis elaborated in Section~\ref{sec:algorithm design} that VWPFA does not damage the base revenue of the whole SAC approach in each individual slot.

\subsubsection{Inter-slice fairness}

\begin{figure*}[!ht]
	\centering
	\subfloat[$\lambda^G = 2$]{
		\includegraphics[width=0.33\linewidth]{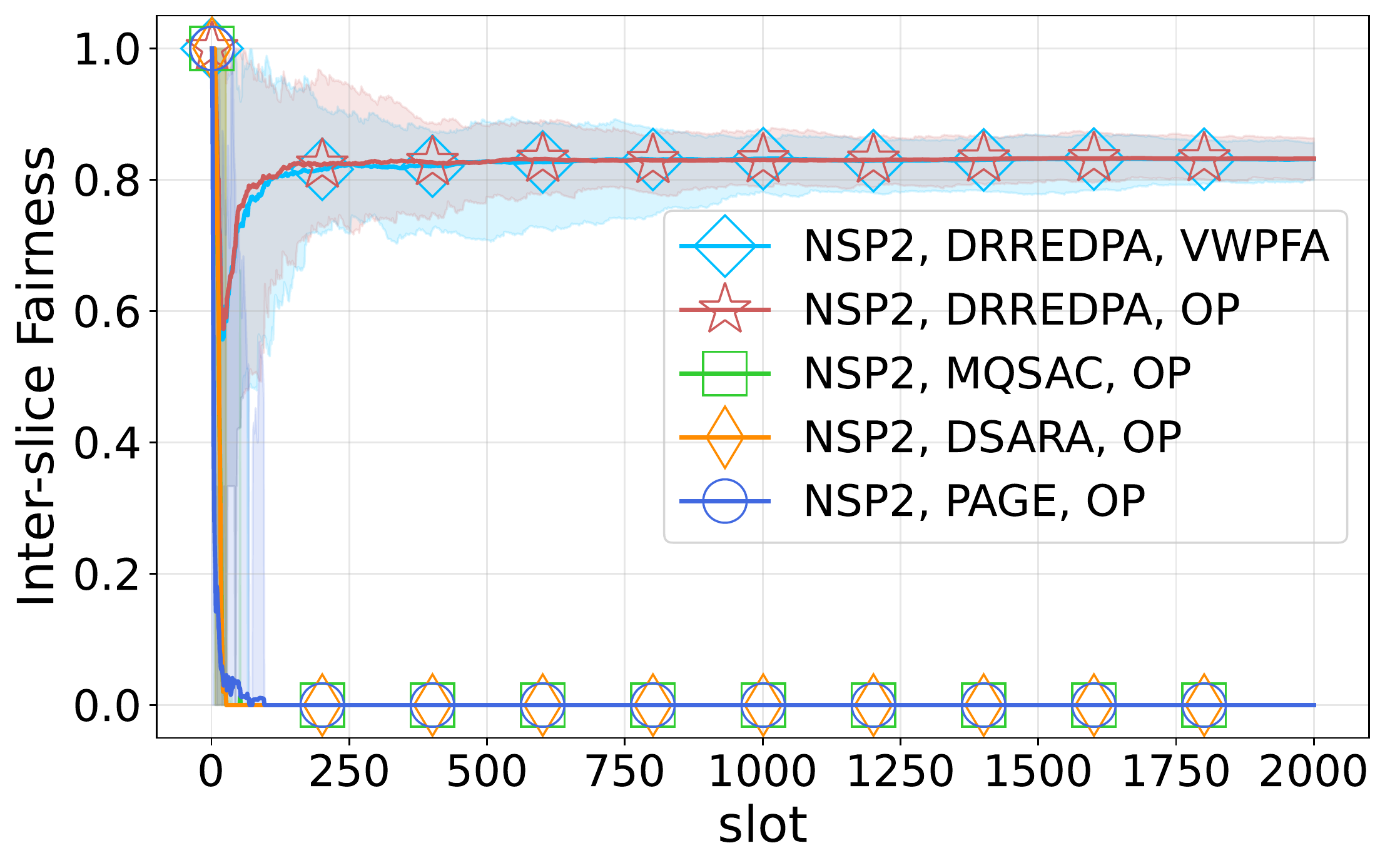}
		\label{fig:inter-slice fairness lambda 2}
	}
	\subfloat[$\lambda^G = 3$]{
		\includegraphics[width=0.33\linewidth]{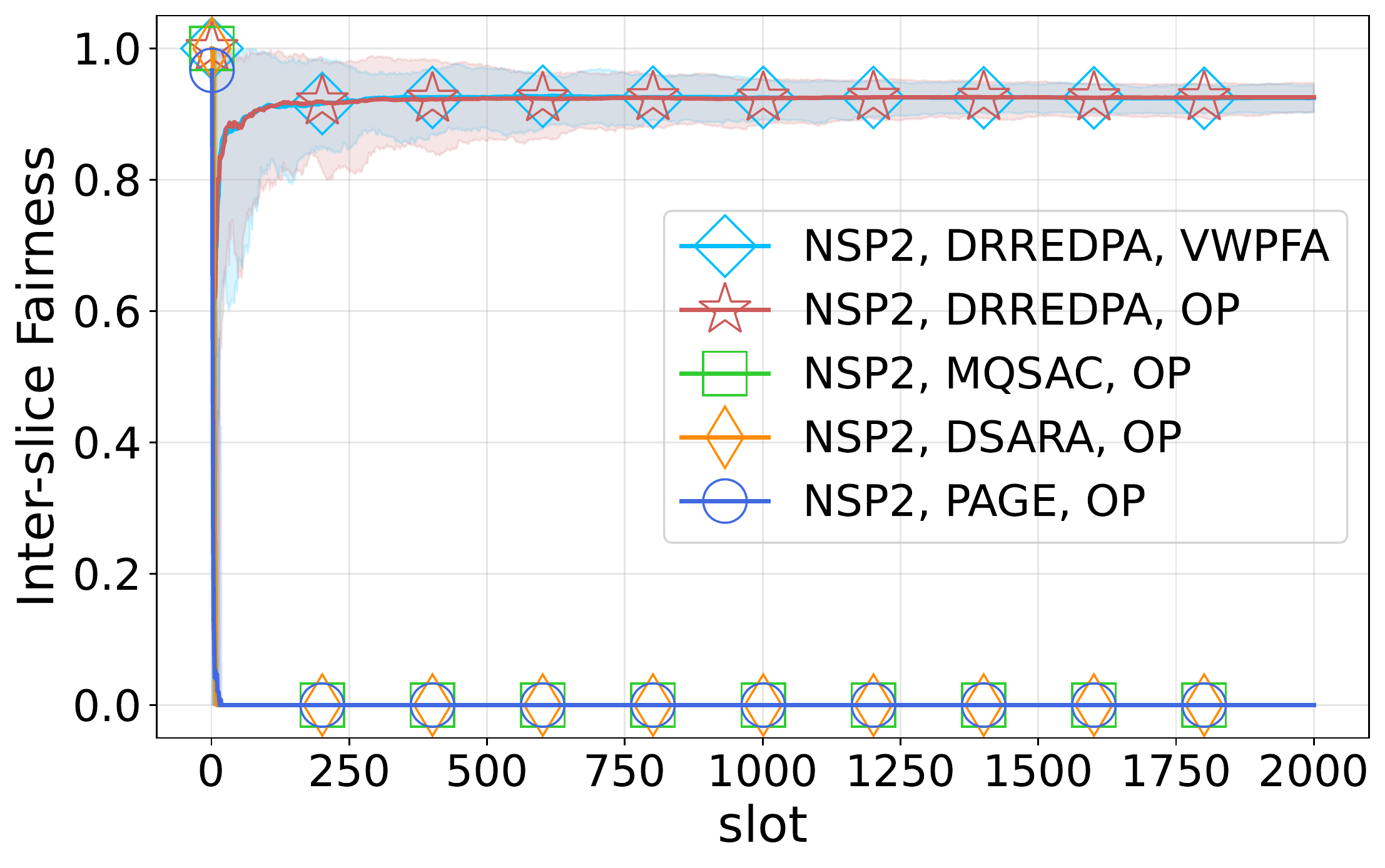}
		\label{fig:inter-slice fairness lambda 3}
	}
	\subfloat[$\lambda^G = 4$]{
		\includegraphics[width=0.33\linewidth]{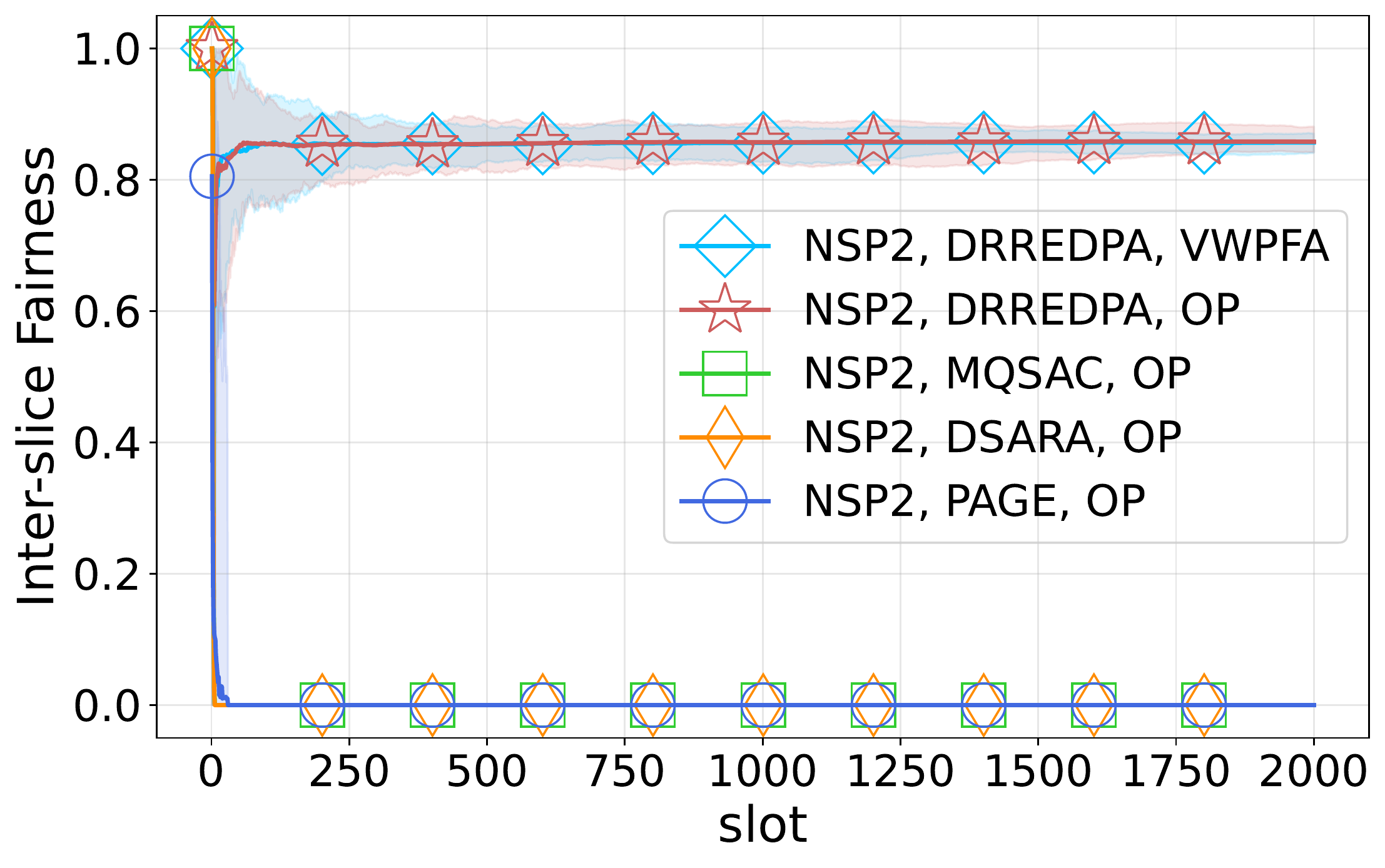}
		\label{fig:inter-slice fairness lambda_4}
	}
	\caption{Inter-slice fairness}
	\label{fig:inter-slice fairness}
\end{figure*}

In this study, the service priority constraint is defined on the average acceptance ratio of adjacent priority slices, listing the ratio curve of each slice under all algorithms in one figure makes it bothering for readers to distinguish whether the ratio increases with priority. Fortunately, when the priority condition is satisfied, the inter-slice fairness index defined by \eqref{eq:inter slice fairness} takes a positive value and 0 otherwise, which provides an explicit way to judge the status of the priority condition.

The inter-slice fairness is presented in \figref{fig:inter-slice fairness}. The cumulative acceptance ratio of all slices is 0 before receiving any slice request, so the inter-slice fairness is initiated to 1. The calculation method of the cumulative acceptance ratio determines that it is sensitive to the new access volume when the amount of requests perceived is small. Therefore, each algorithm experienced a large shock in the early stage, but algorithms with DRREDPA will gradually stabilize to a higher value, while all comparison methods quickly drop to 0. This agrees with the design philosophy of MQSAC and DSARA; they only care about revenue and do not consider maintaining the priority differences of heterogeneous slices. Although PAGE uses the static resource allocation ratio to distinguish different slices, it does not take into account the fact that the resource consumption and request volume of these slices may vary dramatically. In our settings, the third slice provides the most popular service type, with the highest total number of requests from subscribers, but its priority is only in the middle among all slices. The resources reserved by PAGE for slice~3 are clearly not enough to carry its total resource requirements, which makes the algorithm unable to maintain the priority constraints well.

It is worth mentioning that the shaded areas of the two methods using DRREDPA in \figref{fig:inter-slice fairness} do not completely overlap, that is because although VWPFA does not change the average base revenue in a single slot, its decision will affect the queue length of VSPs in the future, and then affect choices of subscribers. Therefore, in the long run, NSP~2 goes through a slightly different state when OP is replaced by VWPFA. However, it turns out that VWPFA hardly causes a noticeable impact on the performance of DRREDPA, which also means that VWPFA can basically be transparent between slices even used for a long-term process.

A higher inter-slice fairness metric not only means that DRREDPA has successfully maintained the hierarchical relationship between services, but also helps attract vertical service providers to establish closer cooperation with network slice providers, thereby obtaining greater revenue margins, which has already been corroborated in \figref{fig:long-term average base revenue}.

\subsubsection{Average VWPF}

\begin{figure*}[!ht]
	\centering
	\subfloat[$\lambda^G = 2$]{
		\includegraphics[width=0.33\linewidth]{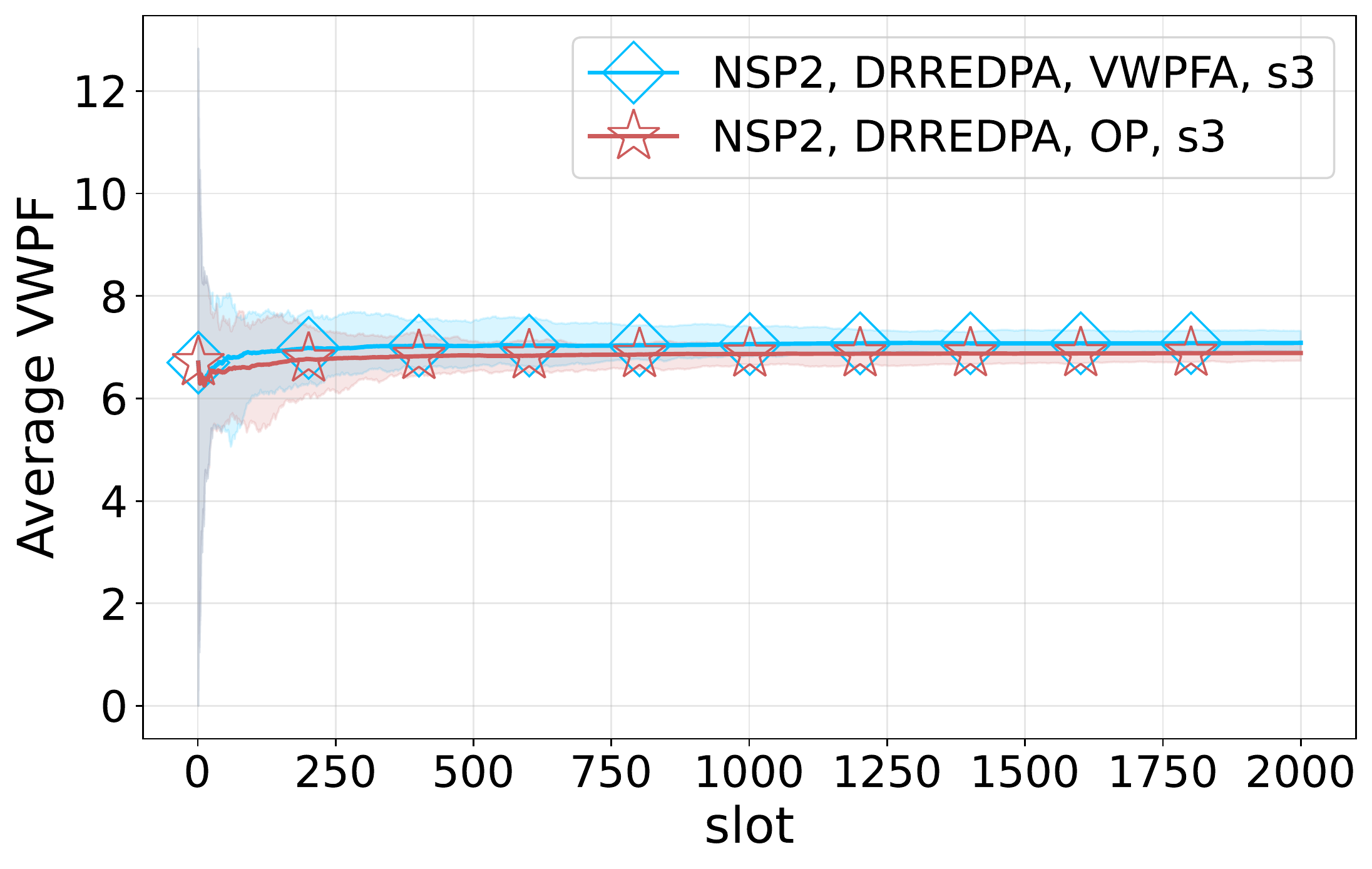}
		\label{fig:avg vwpf lambda 2}
	}
	\subfloat[$\lambda^G = 3$]{
		\includegraphics[width=0.33\linewidth]{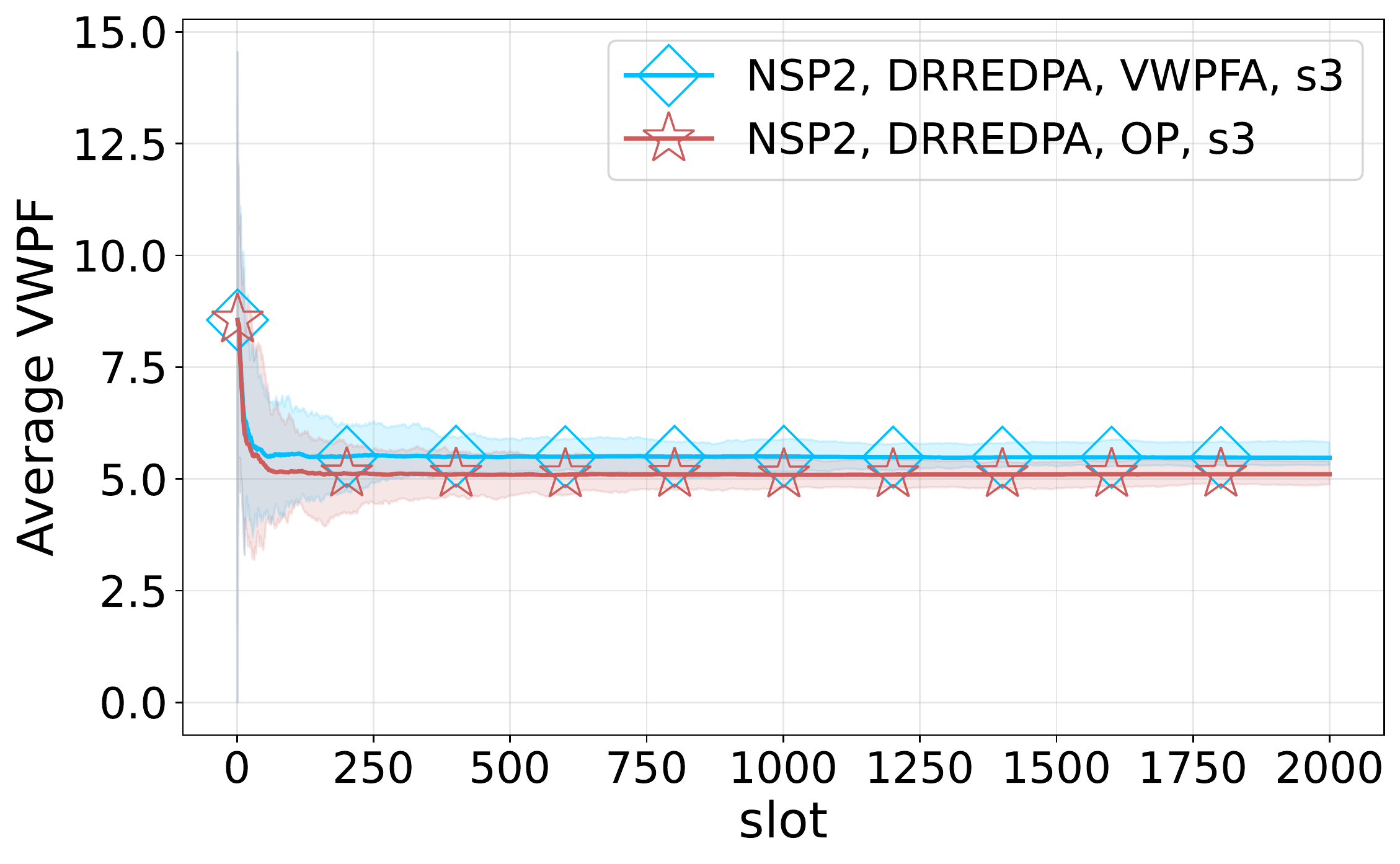}
		\label{fig:avg vwpf lambda 3}
	}
	\subfloat[$\lambda^G = 4$]{
		\includegraphics[width=0.33\linewidth]{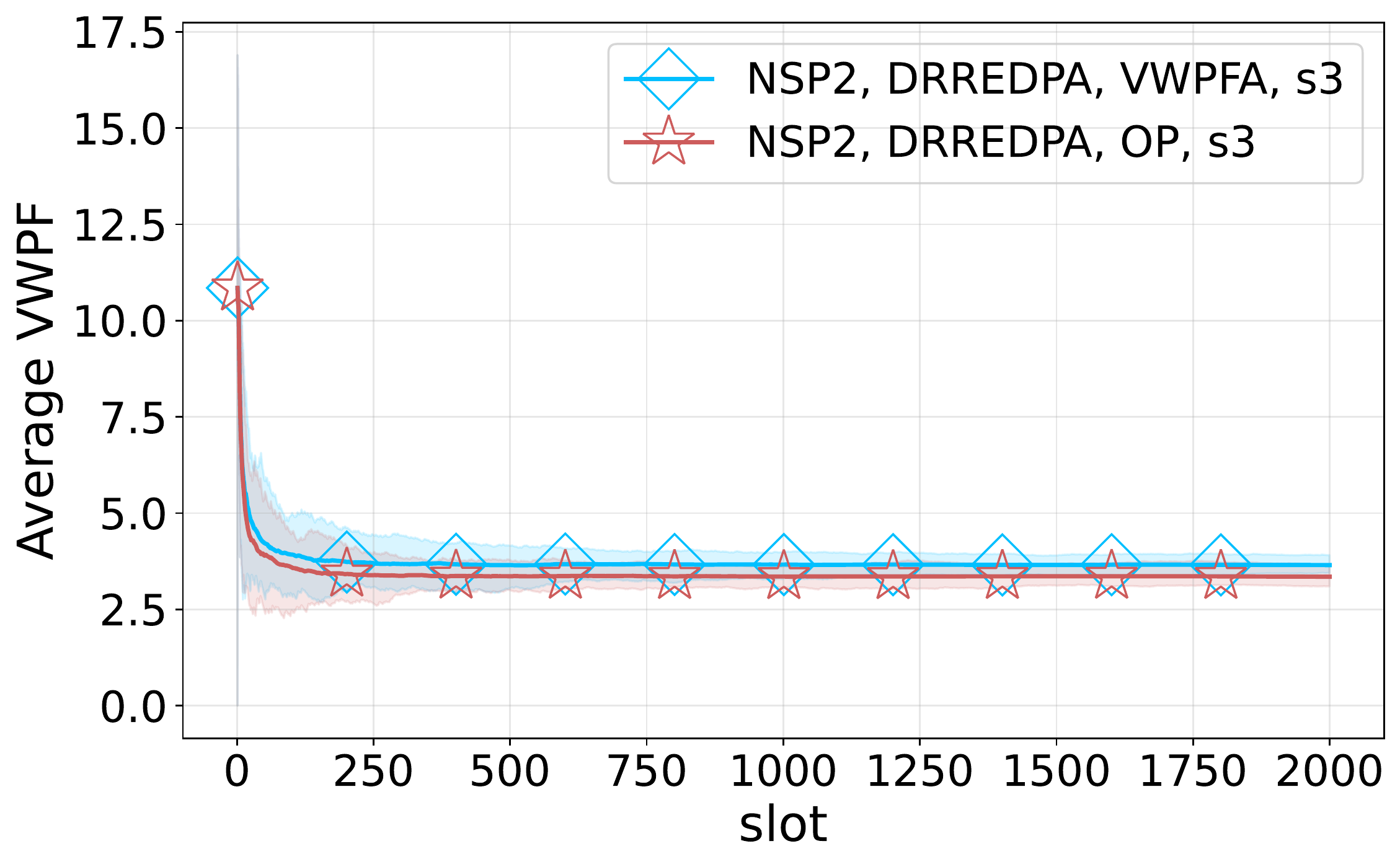}
		\label{fig:avg vwpf lambda 4}
	}
	\caption{Average value-weighted proportional fairness}
	\label{fig:avg vwpf}
\end{figure*}

\figref{fig:avg vwpf} shows the average VWPF of the third slice when the inter-slice algorithm is fixed to DRREDPA and the intra-slice algorithm is either VWPFA or OP. It can be seen that the curve of VWPFA is always higher than that of OP, and the bounds of the shadow also always exceed the counterparts belonging to OP. The improvement seems small because only two VSPs are assumed as competitors for slice~3, and \eqref{eq:AVWPF} is a submodular function. With the increase of the allocation amount, the function increment shows a marginal effect. When $\lambda^G$ gets larger, requests for slice~4 and 5 increase, which are more privileged and resource-consuming than slice~3. Therefore, the algorithm compresses a part of the quotas of slice~3, leading to a decrease in the absolute value of average VWPF.

\subsubsection{Average actual revenue}

\begin{figure*}[!ht]
	\centering
	\subfloat[$\lambda^G = 2$]{
		\includegraphics[width=0.33\linewidth]{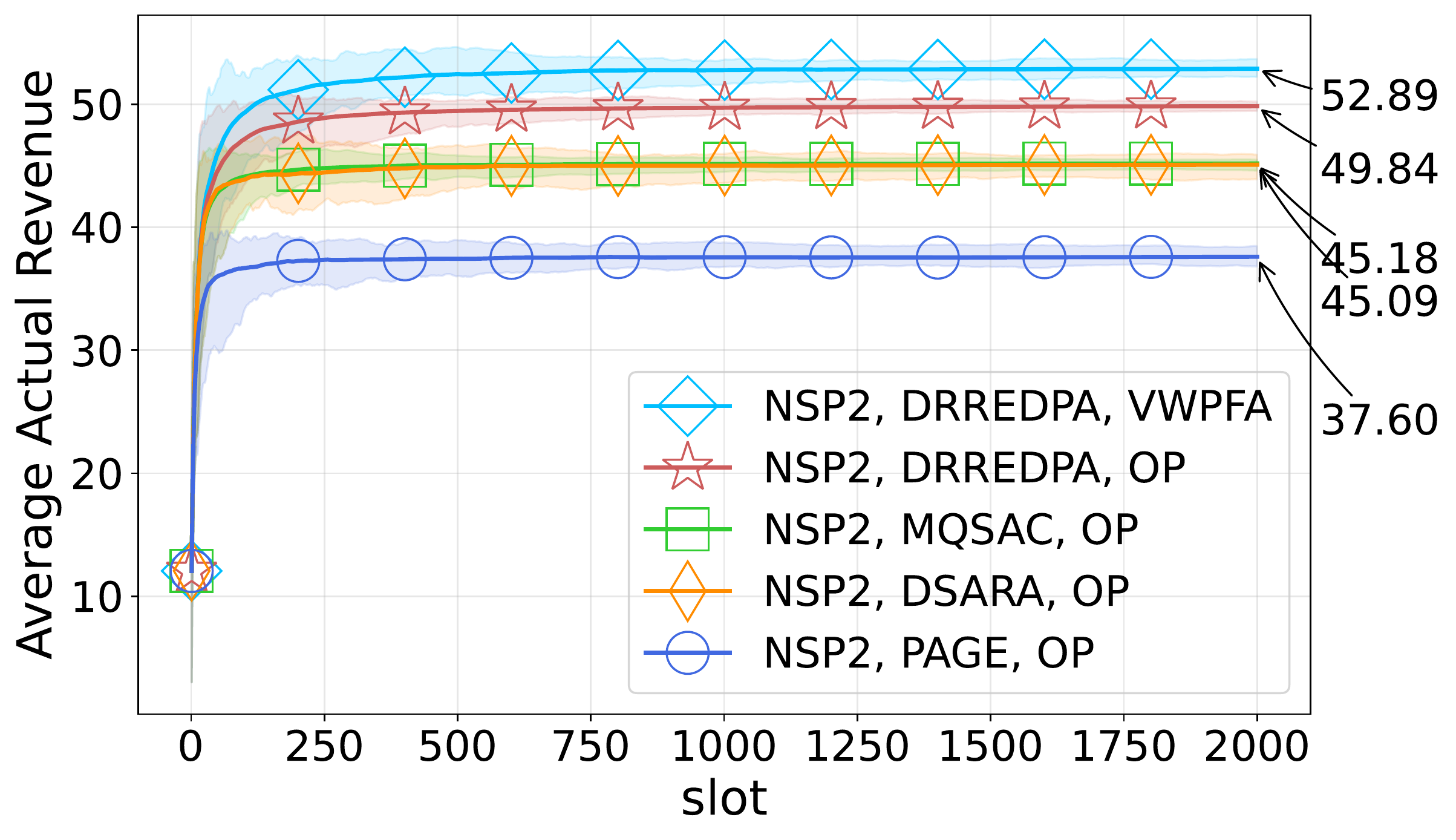}
		\label{fig:long-term average actual revenue lambda 2}
	}
	\subfloat[$\lambda^G = 3$]{
		\includegraphics[width=0.33\linewidth]{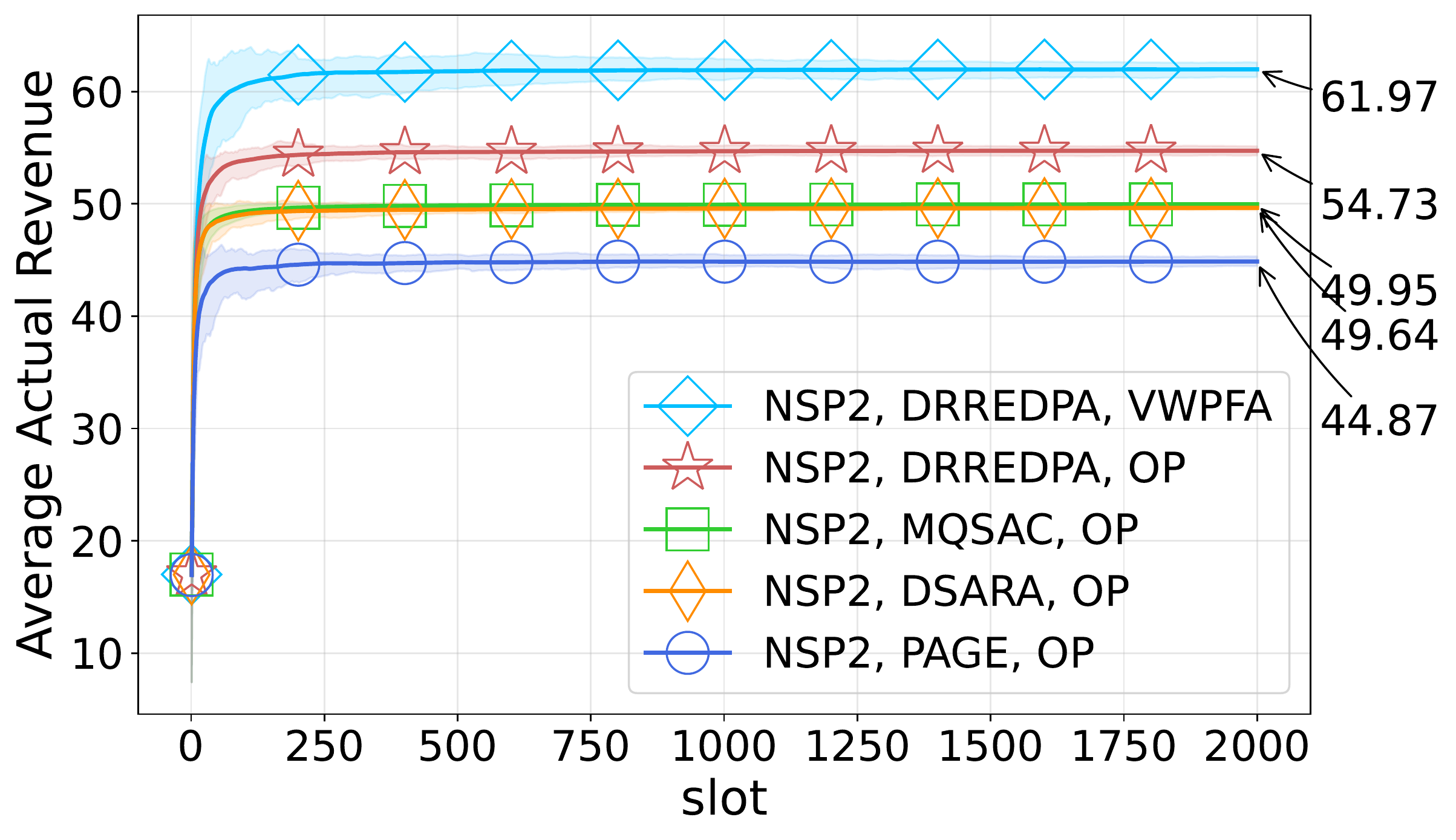}
		\label{fig:long-term average actual revenue lambda 3}
	}
	\subfloat[$\lambda^G = 4$]{
		\includegraphics[width=0.33\linewidth]{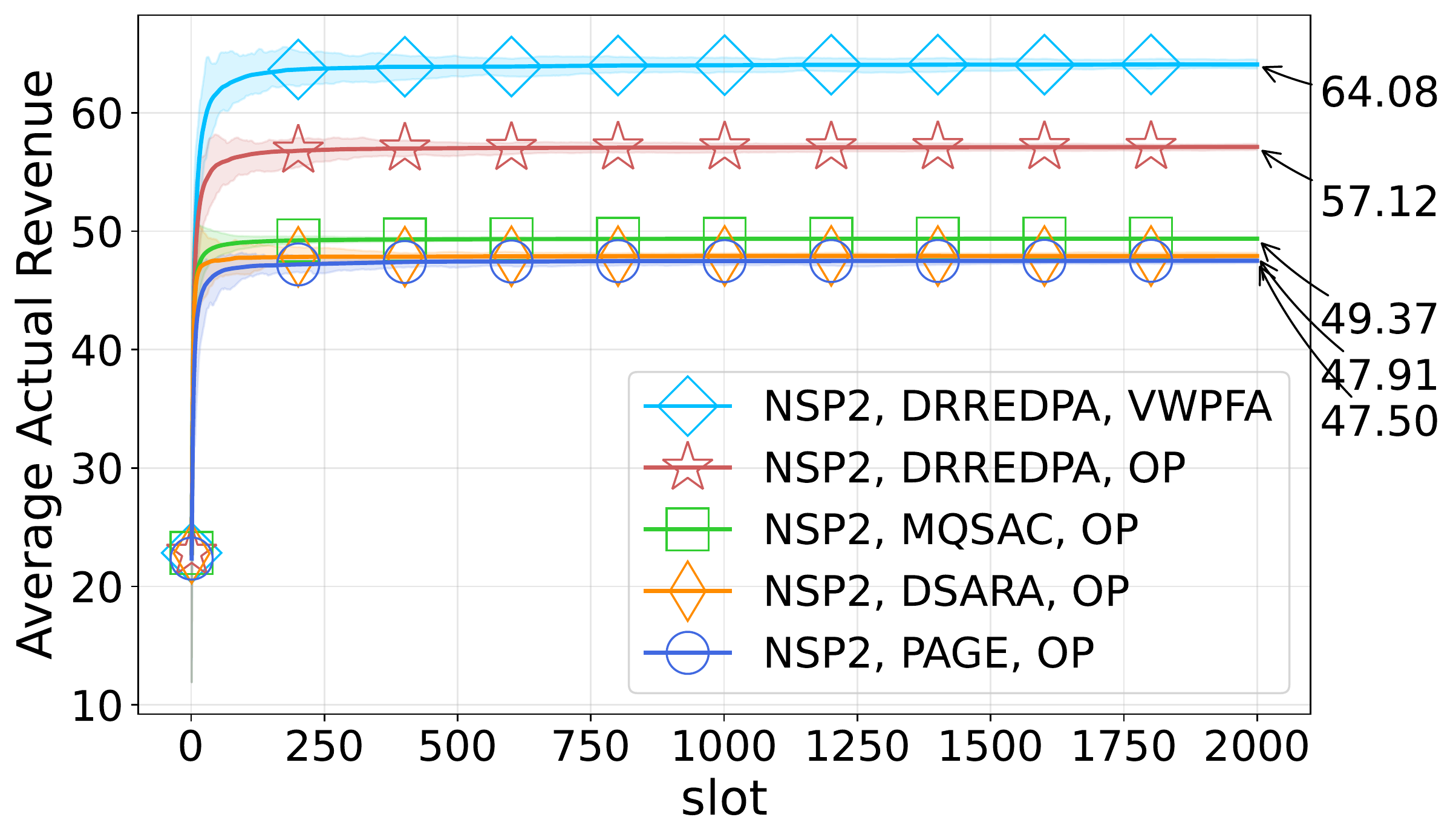}
		\label{fig:long-term average actual revenue lambda 4}
	}
	\caption{Long-term average actual revenue}
	\label{fig:long-term average actual revenue}
\end{figure*}

We have also evaluated the long-term average actual revenue for different algorithms, the results are depicted in \figref{fig:long-term average actual revenue}. Except for MPSAC (i.e., DRREDPA with VWPFA), the long-term average actual revenue of others is the same as their long-term average base revenue shown in \figref{fig:long-term average base revenue}. The OP method has no ability to force VSPs to disclose their true valuation, so NSP~2 can only charge the base price, so as to prevent the strategic behavior of VSPs from seriously damaging its income and avoid charging tenants more than they can afford. Things become different for VWPFA, the designed allocation rule requires VSPs to give their valuations, while the charging rule restricts bidding honestly as the only dominant strategy for VSPs. As long as the bids received are credible, NSP~2 can price the admission quota more reasonably, ensuring that it will not be lower than the reserved price nor exceed the true valuation of VSPs. It can be seen from \figref{fig:long-term average actual revenue} that the auction mechanism we designed has indeed further promoted the actual income of the NSP.

Combining \figref{fig:long-term average base revenue} to \figref{fig:long-term average actual revenue}, we conclude that VWPFA will not interfere with the admission decision between slices, can achieve fair quota allocation within a slice, and has the potential to further enhance the revenue of NSPs.

\subsubsection{Average time cost}

\begin{figure}[!th]
	\centering
	\includegraphics[width=0.66\linewidth]{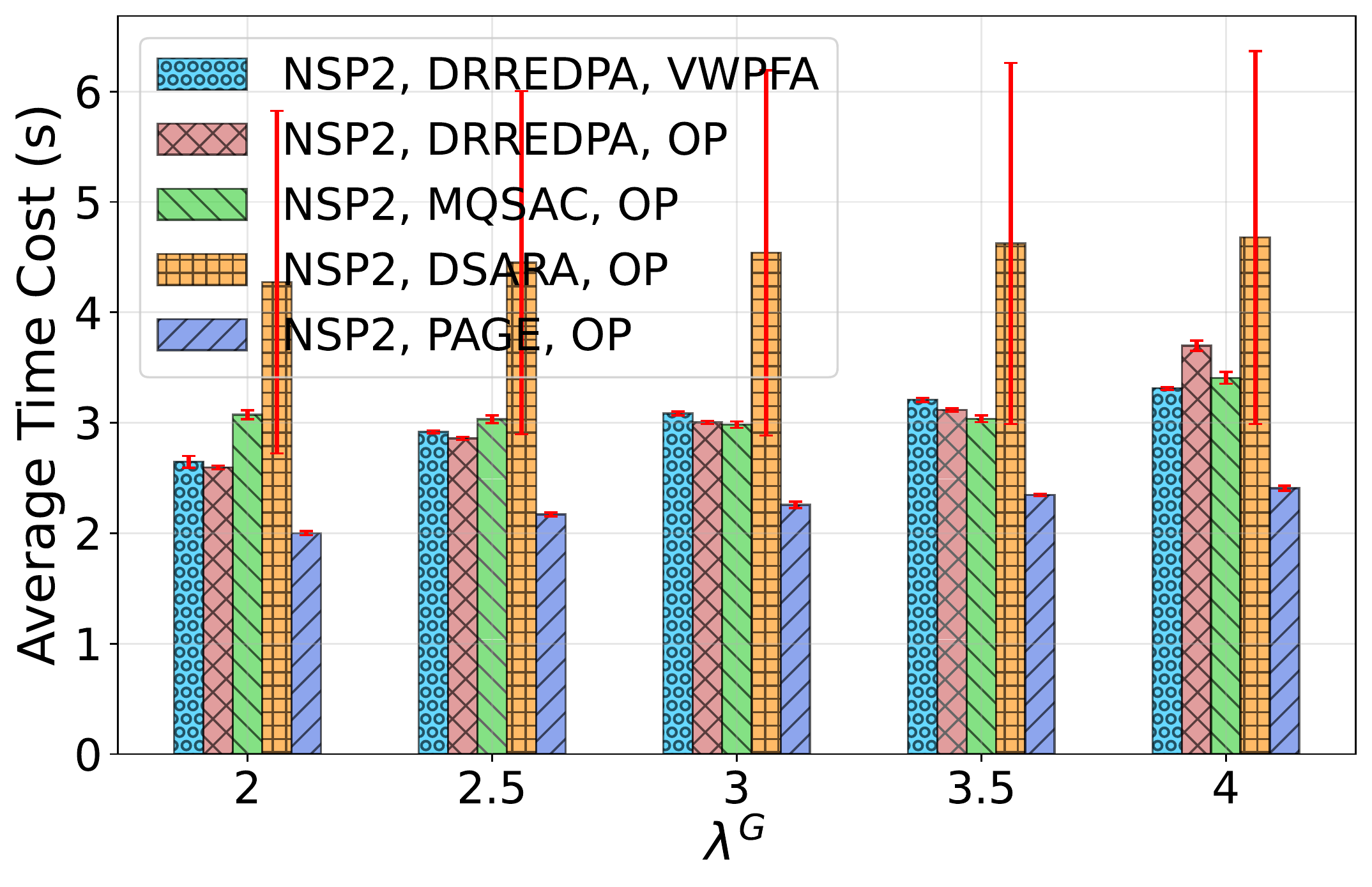}
	\caption{Time consumption}
	\label{fig:time cost}
\end{figure}

Finally, we measure the overall time consumption of each algorithm for 2000 slots, \figref{fig:time cost} presents the mean of 50 repeated statistics. PAGE only considers the least constraints, so its running time is always the shortest. Those of DRREDPA with VWPFA, DRREDPA with OP, and MQSAC with OP are relatively close, the first two are sometimes slightly higher than the third one because DRREDPA needs to frequently calculate the dominant resource efficiencies and sort them. MPSAC is almost always slightly slower than DRREDPA with OP, due to the operation of calculating increments and pricing quotas during the auction process. DSARA takes the longest period for the essential model training in the early stage, so it is worse than others in terms of both average time cost and error deviation.

Although MPSAC we proposed produces some additional overhead in running time compared to some existing algorithms, the total amount is still maintained at a very low level and will not cause timeliness problems. Moreover, its time overhead will not increase rapidly with the increase in service pressure. Considering the advantages of MPSAC, these small sacrifices are totally worth it.

\section{Conclusion and Future Works} \label{sec:conclusion and future works}

Network slicing technology is an important enabler for network operators to implement diversified and differentiated services on generic facilities. However, the research on slice admission control is not deep enough, especially in the scenario where multiple rational participants coexist.

In this paper, we have proposed MPSAC for maximizing the long-term average revenue of NSPs, with the constraints of multidimensional resource feasibility and slice priority relationship. We theoretically proved that the problem is intractable in polynomial time, then designed an inter-slice admission method based on dominant resource revenue efficiency to adequately and preferentially allocate resources, and devised an auction mechanism with DSIC property for intra-slice quota assignment to curb malicious fraud of slice tenants. The results show that MPSAC achieved (at least) up to 9.6\%, 10.3\%, and 20.3\% greater long-term average base revenue than those obtained by MQSAC, DSARA, and PAGE, and it increases to 17.1\%, 17.3\%, and 34.9\% in terms of long-term average actual revenue, benefiting from the capability of the proposed auction mechanism to reveal the true valuations of VSPs. Furthermore, MPSAC can well maintain the priority requirements between slices, while other comparisons cannot, and the allocation rule of the devised auction mechanism achieves the value-weighted proportional fairness within slices. In addition, the extra time overhead brought by these performance improvements is quite limited. These results collectively demonstrate the significance and value of our algorithm.

In the future, we plan to explore the potentiality of Multi-agent Deep Reinforcement Learning (MADRL) on the problem considered in this work. Compared with a single agent, MADRL can solve the non-static issue of the environment observed by individual agents through expanding the state dimension to include observations of other agents. It means that MADRL is expected to fundamentally improve the not-so-good performance of DSARA revealed in \secref{sec:performance evaluation}. However, in a non-cooperative scenario, how can an individual agent obtain the observations of others and how to achieve efficient and high-quality convergence of multiple neural network models still need and deserve further investigations.

\bibliographystyle{IEEEtran}
\bibliography{ref}

\end{document}